\documentclass[preprint,12pt]{elsarticle}

\usepackage{amsmath,amssymb,amsfonts}
\usepackage{amsthm}

\usepackage{multirow}
\usepackage{graphicx}
\usepackage{float}

\usepackage[ruled]{algorithm}
\usepackage{algpseudocode} 

\usepackage{mathtools}
\usepackage{array}
\usepackage{textcomp}
\usepackage{xspace}
\usepackage{color}
\usepackage{url}

\graphicspath{{./Pics/}}

\usepackage{hyperref}
\hypersetup{
  colorlinks=true,
  linkcolor=black,
  urlcolor=blue,
  citecolor=blue,
}

\newtheorem{theorem}{Theorem}
\newtheorem{lemma}{Lemma}

\theoremstyle{definition}

\newtheorem{assumption}{Assumption}

\theoremstyle{remark}
\newtheorem{remark}{Remark}

\journal{Journal of Energy Storage}
\begin{document}
\begin{frontmatter}



\title{Privacy-Preserving Distributed Control for a Networked Battery Energy Storage System}



\author[label1]{Mihitha Maithripala}
\author[label1]{Zongli Lin}

\affiliation[label1]{organization={Charles L. Brown Department of Electrical and Computer Engineering},
            addressline={University of Virginia}, 
            city={Charlottesville},
            postcode={22904}, 
            state={VA},
            country={U.S.A.}}


\begin{abstract}
The increasing deployment of distributed Battery Energy Storage Systems (BESSs) in modern power grids necessitates effective coordination strategies to ensure state-of-charge (SoC) balancing and accurate power delivery. While distributed control frameworks offer scalability and resilience, they also raise significant privacy concerns due to the need for inter-agent information exchange. This paper presents a novel privacy-preserving distributed control algorithm for SoC balancing in a networked BESS. The proposed framework includes a distributed power allocation law that is designed based on two privacy-preserving distributed estimators, one for the average unit state and the other for the average desired power. The average unit state estimator is designed via the state-decomposition method without disclosing sensitive internal states. The proposed power allocation law based on these estimators ensures asymptotic SoC balancing and global power delivery while safeguarding agent privacy from external eavesdroppers. The effectiveness and privacy-preserving properties of the proposed control strategy are demonstrated through simulation results.

\end{abstract}




\end{frontmatter}

\allowdisplaybreaks
	
\section{Introduction}\label{sec:introduction}
Battery Energy Storage Systems (BESSs) are increasingly essential for the operation of modern smart grids due to their ability to buffer the variability of renewable energy, support grid stability, and improve energy efficiency \cite{yu2025optimized, hill2012battery, lawder2014battery}. In most cases, a BESS consists of several battery units located at different sites, forming a multi-agent network. The distributed deployment of BESSs on a household, industrial, or community scale allows localized control, resilience to disturbances, and flexible energy management. However, such systems introduce challenges in the coordination of energy storage resources while maintaining operational goals such as state-of-charge (SoC) balancing and total power delivery/tracking~\cite{meng2021distributed, xu2018distributed}. SoC balancing is crucial for extending battery life, ensuring uniform aging across units, and preventing overcharge or deep discharge conditions~\cite{qiu2025distributed}.

Centralized control approaches rely on global knowledge of all unit states, which limits scalability, increases communication overhead, and is prone to a single point of failure~\cite{liu2025fixed}. In contrast, distributed control methods, particularly those based on dynamic consensus or distributed optimization, have been extensively explored for scalable and reliable SoC balancing~\cite{qiu2025distributed,xing2019distributed, qian2023distributed, lu2014state, huang2023voltage}. These methods enable agents to coordinate through local communication with their neighbors, supporting robust performance even under partial connectivity or information loss.

Numerous decentralized control strategies have emerged in response to the growing complexity of power system applications, each tailored to specific objectives such as voltage regulation, power sharing, or balanced SoC. Many of these approaches rely on consensus convergence techniques and graph-theoretic connectivity to achieve coordination across the network~\cite{qu2018cooperative, bidram2012secondary, yan2019event}. As the implementation of distributed BESS control gains momentum, addressing diverse objectives under varying communication and operational constraints remains a key area of focus.

Despite their advantages, distributed control frameworks raise serious concerns about privacy. Many algorithms require agents to share internal information, such as SoC levels, power references, voltage values, or other state variables, with neighboring agents. This opens the door to inference attacks that can compromise the privacy of participating units~\cite{kossek2024privacy,luan2024privacy}. 

Two common adversarial models emerge in this setting. The first involves external or internal eavesdroppers who monitor communication channels to capture exchanged data and infer sensitive internal states or usage patterns. External eavesdroppers are typically malicious third parties outside the system, while internal eavesdroppers may include compromised nodes within the network that have been taken over by an attacker. The second model involves honest-but-curious agents, who follow the protocol correctly but attempt to infer private information about their neighbors using the data received during coordination~\cite{kossek2024privacy,kia2015dynamic}. For SoC balancing in a BESS example, given enough communication samples, an attacker could reconstruct the trajectory of a neighboring unit’s SoC or estimate its power output. These threats are particularly relevant in collaborative, yet competitive environments, such as peer-to-peer energy trading platforms or energy communities, where participants aim to maximize their own benefit without disclosing operational strategies~\cite{kossek2024privacy, zhang2022privacy}.

Several approaches have been developed to address these privacy threats. Encryption-based methods, such as homomorphic encryption, enable computations on encrypted data without requiring prior decryption. This allows internal agents to process the data without accessing the underlying information. However, these methods are computationally intensive and often unsuitable for real-time applications~\cite{7852360, lu2018privacy}. Differential privacy techniques address privacy by injecting calibrated statistical noise into shared data, providing rigorous privacy guarantees. Nevertheless, they introduce trade-offs between privacy levels, accuracy, and convergence speed~\cite{Nozari2015DifferentiallyPA, ding2021differentially, 9910413}. As an alternative, methods such as state-decomposition have gained attention for their ability to preserve privacy while maintaining algorithmic performance~\cite{luan2024privacy,zhang2022privacy, 8657789}. This approach partitions the state of each agent into an observable (shared) component and a hidden (private) component. Only the observable part participates in distributed computation, whereas the private component remains hidden to safeguard sensitive information, preventing adversaries from inferring private information through shared data.

State-decomposition enables accurate algorithmic convergence while inherently shielding sensitive information from exposure. It has been successfully applied in dynamic average consensus and static average consensus~\cite{zhang2022privacy, 8657789}, and is now gaining interest in distributed optimization of energy systems~\cite{kossek2024privacy,luan2024privacy}. By embedding privacy into the control architecture itself, these methods inherently protect information without sacrificing performance that aligns well with distributed real-time control requirements.

In this work, we propose a novel privacy-preserving distributed control framework for SoC balancing in networked BESSs. We design a power allocation algorithm based on two estimators, the average unit state estimator and the average desired power estimator. Each battery unit independently estimates the average unit state and average desired power, and contributes to local power allocation. The proposed method enables accurate SoC balancing and power tracking. The algorithm is lightweight, scalable, and well suited for implementation over an undirected communication graph network.

The contributions of the paper can be summarized as follows,
\begin{itemize}

    \item We develop a privacy-preserving dynamic average consensus-based estimator that leverages state-decomposition to protect the privacy of both individual battery internal states and their average. The proposed state-decomposed estimator safeguards the internal states of individual batteries against external eavesdroppers. Additionally, the algorithm enhances privacy by ensuring that each battery's estimated average converges to a predetermined scaled version of the true average.

    \item We design a privacy-preserving distributed estimator for the average desired power, which achieves consensus over a scaled version of the true average. This approach prevents external observers from inferring global power demand while allowing each battery to recover the correct value locally.

    \item We propose a distributed power allocation law for BESSs that achieves SoC balancing and desired power tracking using the above privacy-preserving distributed estimators.

    \item We provide a theoretical analysis establishing convergence, stability, and privacy guarantees of the proposed framework and demonstrate its effectiveness through numerical simulations.

\end{itemize}

The remainder of the paper is structured as follows. Section~\ref{sec:preliminaries} provides the preliminaries, including a description of the communication network and an attack model relevant to dynamic average consensus algorithm. Section~\ref{sec:problem} formulates the baseline control problem for coordinated SoC balancing and power tracking without privacy considerations. Section~\ref{sec:privacy} introduces the proposed privacy-preserving distributed control algorithm, which incorporates estimators for both the average unit state and the average desired power. Section~\ref{sec:simulation} presents simulation results that validate the effectiveness and privacy protection achieved by the proposed strategy. Finally, Section~\ref{sec:conclusion} concludes the paper.
\begin{figure}[!t]
\centering
\includegraphics[width=5in]{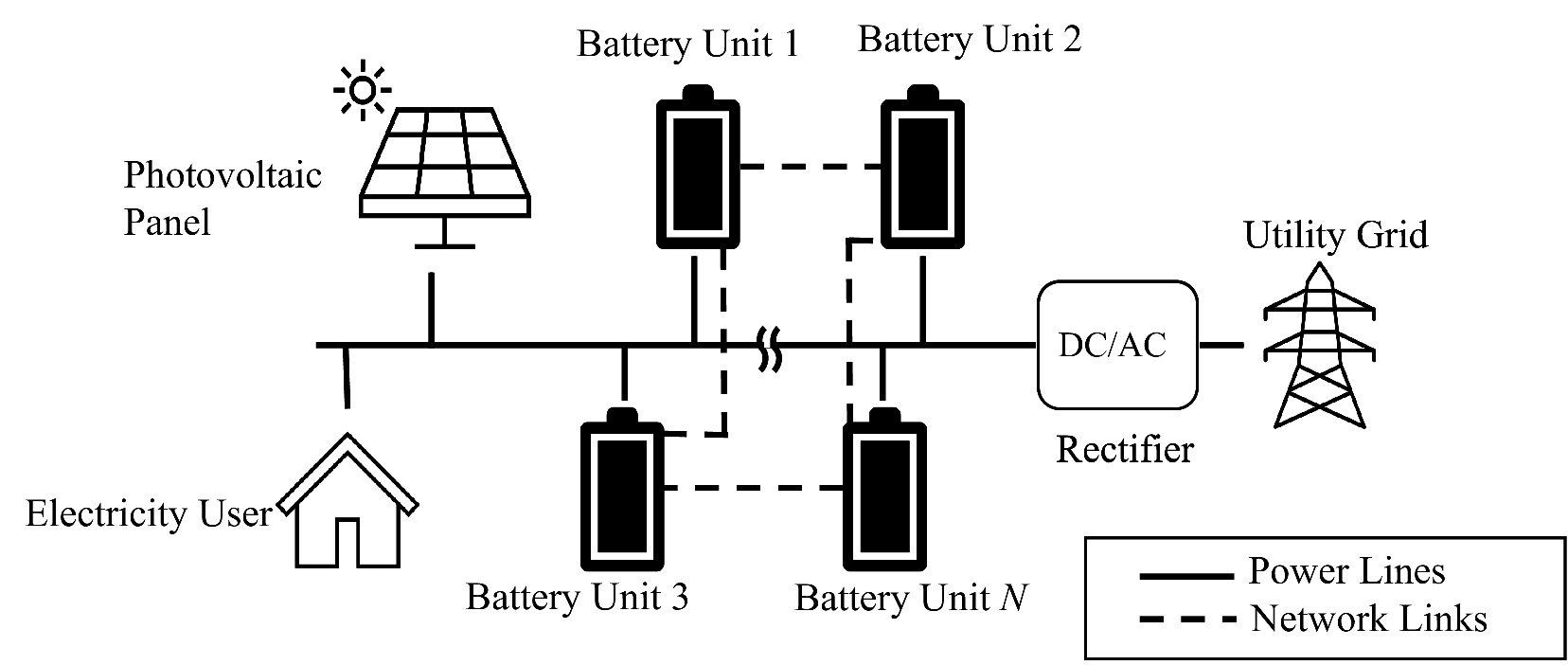}
\caption{The illustration of a microgrid including a BESSs.}
\label{fig:fig1}
\end{figure}

\section{Preliminaries}\label{sec:preliminaries}
\subsection{Communication Network}
We consider a BESS consisting of $N$ networked battery units. Figure \ref{fig:fig1} provides an illustrative example to help to understand the BESS setup in a microgrid. The interactions among the battery units are represented by a graph \( \mathcal{G} = (\mathcal{N}, \mathcal{E}) \), where \( \mathcal{N} = \{1, 2, \dots, N\} \) represents the set of nodes, each corresponding to a battery unit, and \( \mathcal{E} \subseteq \mathcal{N} \times \mathcal{N} \) denotes the set of edges, each representing a communication link between units. If \( (i, j) \in \mathcal{E} \), it means that communication from unit \( i \) to unit \( j \) is permitted. In such a case, unit \( i \) serves as an \textit{in-neighbor} of unit \( j \), while unit \( j \) acts as an \textit{out-neighbor} of unit \( i \). The set of all in-neighbors of unit $i$ is defined as $\mathcal{N}_i = \{j \in \mathcal{N} \mid (j,i) \in \mathcal{E} \}$. A graph \( \mathcal{G} \) is classified as undirected if the presence of an edge \( (i, j) \in \mathcal{E} \) automatically implies that \( (j, i) \in \mathcal{E} \) as well. In the graph \( \mathcal{G} \), a path from node \( i_1 \) to node \( i_k \) is defined by a sequence of directed edges \( \{(i_1, i_2), (i_2, i_3), \dots, (i_{k-1}, i_k)\} \), where the nodes \( i_1, i_2, \dots, i_k \) are all unique. The graph \( \mathcal{G} \) is considered connected if every pair of distinct nodes can be linked through such a path.

The adjacency matrix corresponding to the graph \( \mathcal{G} \) is given by \( A = [a_{ij}] \in \mathbb{R}^{N \times N} \), where each entry \( a_{ij} \) is equal to 1 if there exists a directed edge from unit \( j \) to unit \( i \) (i.e., \( (j, i) \in \mathcal{E} \)) and \( i \neq j \), and \( a_{ij} = 0 \) otherwise. If $\mathcal{G}$ is undirected, then $a_{ij} = a_{ji}$. The Laplacian matrix of the graph \( \mathcal{G} \) is defined as \( L = [l_{ij}] \in \mathbb{R}^{N \times N} \), where \( l_{ij} = -a_{ij} \) for \( i \ne j \), and \(l_{ii} = \sum_{k=1, k \neq i}^{N} a_{ik}\). Its eigenvalues are denoted by \( \lambda_1, \lambda_2, \dots, \lambda_N \). If \( \mathcal{G} \) is connected and undirected, then \( L \) satisfies the following properties \cite{4118472}: \( \lambda_1 = 0 < \lambda_2 \leq \dots \leq \lambda_N \), $\mathbf{1}_N^{\rm T} L = L \mathbf{1}_N = 0$, and for any column vector \( \delta \) with \( \mathbf{1}_N^{\rm T} \delta = 0 \), it holds that \( \delta^{\rm T} L \delta \geq \lambda_2 \|\delta\|^2 \). (Let $\mathbf{1}_N$ denote the column vector with $N$ elements, all equal to $1$).

We impose the following assumption on the communication graph $\mathcal{G}$.

\begin{assumption} \label{ass: graph}
The communication graph \( \mathcal{G} = (\mathcal{N}, \mathcal{E}) \) is undirected and connected.
\end{assumption}

\subsection{Privacy Definition and Attack Model for Dynamic Average Consensus Algorithm}\label{sec:attack}

This work examines a BESS consisting of $N$ battery units that seek to collectively estimate the average of time-varying internal states $x_i(t) \in \mathbb{R}$ using the dynamic average consensus (DAC) algorithm:
\begin{equation}
\dot{\hat{x}}_{{\rm a},i}(t) = \dot{x}_i(t) - \beta \sum_{j=1}^{N} a_{ij} (\hat{x}_{{\rm a},i}(t) - \hat{x}_{{\rm a},j}(t)),
\end{equation}
where \( \hat{x}_{{\rm a},i} \) represents the local estimate of $\frac{1}{N} \sum_{j=1}^{N} x_j(t)$, \( \beta > 0 \) is a design parameter, and $a_{ij} \in \{0,1\}$ represents the communication connexity between battery unit $i$ and $j$. While this method ensures consensus, it requires each battery unit to share it's estimate $\hat{x}_{{\rm a},i}(t)$, which implicitly encodes information about the battery unit state $x_i(t)$ and its derivative $\dot{x}_i(t)$.

In this work, we consider an external eavesdropper who has knowledge of the communication network $(A,\beta)$ and access to all transmitted information $\hat{x}_{{\rm a},i}(t)$.
 The concern is that the eavesdropper can exploit these data to infer a battery unit's private information $x_i(t)$ and $\dot{x}_i(t)$.

To demonstrate the vulnerability of the DAC scheme, we present the observer-based attack model from \cite{zhang2022privacy}. Let variables ${{v}}_i(t)$, $\xi_i(t)$, and $\phi_i(t)$ represent the eavesdropper's reconstructions of $\hat{x}_{{\rm a},i}(t)$, $x_i(t)$, and $\dot{x}_i(t)$, respectively. Then,
\begin{subequations}
\begin{alignat}{3}
\dot{{v}}_i(t) &= {\phi}_i(t) - \beta \sum_{j=1}^{N} a_{ij}(\hat{x}_{{\rm a},i}(t) - \hat{x}_{{\rm a},j}(t)) \notag \\
&\quad + k_1(\hat{x}_{{\rm a},i}(t) - {v}_i(t)), \\
\dot{\xi}_i(t) &= k_2(\hat{x}_{{\rm a},i}(t) - z_i(t) - \xi_i(t)) + \phi_i(t), \\
\phi_i(t) &= k_3 \hat{x}_{{\rm a},i}(t) + \phi'_i(t), \\
\dot\phi'_i(t) &= -k_3 \left(\phi_i(t) - \beta \sum_{j=1}^{N} a_{ij}(\hat{x}_{{\rm a},i}(t) - \hat{x}_{{\rm a},j}(t)) \right)  \notag \\
  &\quad +k_4(\hat{x}_{{\rm a},i}(t) - {v}_i(t)), \\
\dot{z}_i(t) &= -\beta \sum_{j=1}^{N} a_{ij}(\hat{x}_{{\rm a},i}(t) - \hat{x}_{{\rm a},j}(t)), \quad z_i(0) = 0,
\end{alignat}
\end{subequations}
where \( k_1, k_2, k_3, k_4 \in \mathbb{R} \) are positive design parameters, and \( \hat{\phi}'_i(t)\) and \( z_i(t)\) are auxiliary variables.

Under the assumption that $x_i(t), \dot{x}_i(t), \ddot{x}_i(t) \in \mathcal{L}_\infty$, it has been shown \cite{zhang2022privacy} that this observer guarantees uniformly ultimately bounded (UUB) estimation errors, and asymptotic recovery is possible if $\ddot{x}_i(t) \in \mathcal{L}_2$. This confirms that the conventional DAC algorithm is vulnerable to privacy leakage unless privacy-preserving mechanisms are incorporated.

\section{Problem Formulation}\label{sec:problem}

This study focuses on a distributed battery network comprising \( N \) coordinated units, where each unit is indexed by $i \in \mathcal{N} = \{1,2,\dots,N\}$. The SoC dynamics of each energy storage device are modeled using the Coulomb counting approach. Accordingly, the SoC of the \( i \)th battery unit at time \( t \geq 0 \) is given by~\cite{xing2019distributed}
\begin{equation}
S_i(t) = S_i(0) - \frac{1}{C_i} \int_{0}^{t} i_i(\tau) d\tau, \quad t \geq 0,
 \label{eq:1}
\end{equation}
where \( S_i(0) \) represents the initial SoC, \( C_i \) denotes the battery capacity, and the output current is represented by  \( i_i(t) \). The sign of \( i_i(t) \) determines the operational mode of the battery: a positive current (\( i_i(t) > 0 \)) indicates discharging, while a negative current (\( i_i(t) < 0 \)) signifies charging. Differentiating equation \eqref{eq:1} with respect to time results in:
\begin{equation}
\dot{S}_i(t) = -\frac{1}{C_i} i_i(t).
 \label{eq:2}
\end{equation}

The output power \( p_i(t) \) of battery unit \( i \) is expressed as
\begin{equation}
p_i(t) = V_i(t) i_i(t),
 \label{eq:3}
\end{equation}
where \( V_i(t) \) denotes the voltage output of the battery unit. Similar to the current, positive power (\( p_i(t) > 0 \)) corresponds to discharging, whereas negative power (\( p_i(t) < 0 \)) corresponds to charging. 

Bidirectional DC--DC converters in battery energy storage systems are typically designed to regulate a constant output voltage~\cite{inoue2007bidirectional}. Therefore, we simplify the model by assuming that each battery unit maintains a constant output voltage, i.e., \( V_i(t) = V_i \). Under this assumption, substituting equation \eqref{eq:3} into equation \eqref{eq:2} results in
\begin{equation}
\dot{S}_i(t) = -\frac{1}{C_i V_i} p_i(t).
 \label{eq:4}
\end{equation}

This equation describes the dynamics of the SoC, taking into account both its output power and capacity.

The aim of this work is to develop a privacy-preserving distributed power allocation strategy that ensures SoC balancing across battery units while meeting the required power demand in both charge and discharge operations. 

The following are the control objectives for the BESS.

\textbf{Problem 1:} Let the BESS comprise \( N \) interconnected battery units, where the SoC dynamics are governed by equation~\eqref{eq:4}.
The goal is to design distributed control laws for managing charging and discharging power such that:

1) The SoC balancing among the battery units is achieved in steady state with a predefined accuracy \( \epsilon_{\rm s} \geq 0 \), ensuring:
\[
    \lim_{t \to \infty} |S_i(t) - S_j(t)| \leq \epsilon_{\rm s}, \quad  i, j \in \mathcal{N}.
\]

2) In steady state, the total charging/discharging power \( p_\Sigma(t) \) follows the desired power \( p^*(t) \) within a predefined accuracy \( \epsilon_{\rm p} \geq 0 \), i.e.,
\[
\lim_{t \to \infty} |p_{\Sigma}(t) - p^*(t)| \leq \epsilon_{\rm p},
\]
where the total power is defined as:
\[
p_{\Sigma}(t) = \sum_{i=1}^{N} p_i(t).
\]

To solve this problem, we first analyze the power allocation problem in a non-privacy-preserving setting in Section \ref{sec:Non privacy}, establishing the fundamental principles and control laws necessary for effective SoC balancing. Subsequently, in Section \ref{sec:privacy} we will introduce privacy-preserving mechanisms, ensuring that the distributed power allocation strategy maintains confidentiality while still achieving the desired system objectives. 

Before analyzing the non-privacy-preserving case, as explored in References \cite{meng2021distributed}, \cite{xing2019distributed} and \cite{qian2023distributed}, we first establish a set of mild assumptions concerning the total power requirement and the extent to which individual battery units have access to this information.
\begin{assumption} \label{ass: power}
The desired power \( p^*(t) \) is bounded such that \( p \leq |p^*(t)| \leq \bar{p} \) and its rate of change satisfies \( |\dot{p}^*(t)| \leq \psi \), with \( p \), \( \bar{p} \), and \( \psi \) being positive constants.  
\end{assumption}
\begin{assumption} \label{ass: knowledge}
 The desired power \( p^*(t) \) is known by at least one battery unit in the BESS.
\end{assumption}

\subsection{Power Allocation Law Based on Average Estimators in a Non-Privacy-Preserving Setting} \label{sec:Non privacy}

This section revisits the distributed power allocation algorithms introduced in \cite{meng2021distributed} and introduces a new battery average unit state estimator using the conventional dynamic average consensus algorithm.

The state of battery unit \( i \in \mathcal{N} \) is defined to facillitate power allocation as
\[
x_i(t) =
\begin{cases} 
x_{{\rm d},i}(t) = C_i V_i S_i(t), & \text{(discharging mode)}, \\
x_{{\rm c},i}(t) = C_i V_i (1 - S_i(t)), & \text{(charging mode)}.
\end{cases}
\]

Here, \( x_i(t) \) denotes the amount of electrical energy that battery unit \( i \) is capable of storing during charging or delivering during discharging. As a result of physical limitations, positive constants \( a_1 \) and \( a_2 \) exist such that
\begin{equation}
a_1 \leq x_i(t) \leq a_2, \quad  t \geq 0, \quad  i \in \mathcal{N} .
\label{eq:5}
\end{equation}

From the SoC dynamics in equation \eqref{eq:4}, the expression for \( \dot{x}_i(t) \) can be expressed as
\begin{equation}
\dot{x}_i(t) =\!
\begin{cases} 
\dot{x}_{{\rm d},i}(t) = C_i V_i \dot{S}_i(t) = -p_i(t), &  \!\!\text{(discharging mode)}, \\
\dot{x}_{{\rm c},i}(t) = -C_i V_i \dot{S}_i(t) = p_i(t), & \!\!\text{(charging mode)}.
\end{cases}
\label{eq:mode}
\end{equation}

The power allocation for the \( i \)th battery unit is governed by the following control law.
\begin{equation}
p_i(t) =
\begin{cases} 
\frac{x_{{\rm d},i}(t)}{\sum_{j=1}^{N} x_{{\rm d},j}(t)} p^*(t), & \text{(discharging mode)}, \\
\frac{x_{{\rm c},i}(t)}{\sum_{j=1}^{N} x_{{\rm c},j}(t)} p^*(t), & \text{(charging mode)}.
\end{cases}
\label{eq:6}
\end{equation}

To simplify the representation, we introduce the average unit state as follows,
\begin{equation}
x_{\rm a}(t) =
\begin{cases} 
\frac{1}{N} \sum_{i=1}^{N} x_{{\rm d},i}(t), & \text{(discharging mode)}, \\
\frac{1}{N} \sum_{i=1}^{N} x_{{\rm c},i}(t), & \text{(charging mode)},
\end{cases}
\label{eq:7}
\end{equation}
and define the average desired power \(p_{\rm a}(t)\) as
\begin{equation}
p_{\rm a}(t) = \frac{1}{N} p^*(t).
\label{eq:8}
\end{equation}

Using these definitions, the power allocation law in~\eqref{eq:6} can be reformulated as
\begin{equation}
p_i(t) = \frac{x_i(t)}{x_{\rm a}(t)} p_{\rm a}(t).
\label{eq:9}
\end{equation}

Individual battery units may not have access to \( x_{\rm a}(t) \) and \( p_{\rm a}(t) \), as these are global parameters within the BESS. Therefore, these values must be locally estimated to enable decentralized implementation.

All individual battery units employ a dynamic average consensus scheme, as proposed in~\cite{spanos2005dynamic}, to construct a distributed estimator for the average unit state,
\begin{equation*}
\dot{\hat{x}}_{{\rm a},i}(t) = \dot{x}_i(t)  - \beta \sum_{j=1}^{N} a_{ij} (\hat{x}_{{\rm a},i}(t) - \hat{x}_{{\rm a},j}(t)), 
\end{equation*}
\begin{equation}
\hat{x}_{{\rm a},i}(0) = x_i(0),
\label{eq:10}
\end{equation}
where the local estimate of \( x_{\rm a}(t) \) is represented by \( \hat{x}_{{\rm a},i} \), and \( \beta > 0 \) is introduced as a design parameter. Note that although equation \eqref{eq:10} tracks \( x_{\rm a}(t) \) with a steady-state error, which can be made arbitrarily small by tuning the value of \( \beta \).

In a similar manner, the following distributed estimator is constructed to estimate the average desired power at battery unit \( i \). 
\begin{equation*}
\dot{\hat{p}}_{{\rm a},i}(t)\!=\!-\kappa \!\left( \sum_{j=1}^{N} a_{ij} (\hat{p}_{{\rm a},i}(t) - \hat{p}_{{\rm a},j}(t))\! + b_i (\hat{p}_{{\rm a},i}(t) - p_{\rm a}(t))\! \!\right)\!\!,
\end{equation*}
\begin{equation}
\hat{p}_{{\rm a},i}(0) = 0,
\label{eq:12}
\end{equation}
where \( \hat{p}_{{\rm a},i} \) denotes the local approximation of the average desired power \( p_{\rm a}(t) \), and \( \kappa > 0 \) represents a design parameter.

Using these estimates, we refine the distributed power allocation strategy given in \eqref{eq:6}. The updated power allocation law given by
\begin{equation}
\!\!p_i(t) =
\begin{cases} 
\frac{x_{{\rm d},i}(t)}{\max{\left\{\! \frac{a_1}{2}, \hat{x}_{{\rm a},i}(t) \!\right\}}} \hat{p}_{{\rm a},i}(t), & \text{(discharging mode)}, \\
\frac{x_{{\rm c},i}(t)}{\max{\left\{ \!\frac{a_1}{2}, \hat{x}_{{\rm a},i}(t) \!\right\}}} \hat{p}_{{\rm a},i}(t), & \text{(charging mode)},
\end{cases}
\label{eq:14}
\end{equation}
where $a_1>0$ is given in \eqref{eq:5}. The term $\frac{a_1}{2}>0$ is added to ensure that the denominator does not become zero, making the power allocation law feasible to implement.

Before proceeding to discuss the results related to distributed power allocation algorithm \eqref{eq:14}, we first review the convergence  analysis of the estimators defined in \eqref{eq:10} and \eqref{eq:12}.

\begin{lemma}\label{lemma1}
\cite{kia2019tutorial} There exists a positive constant \( \gamma_{\rm s} > 0 \) such that, for every \(  \beta > 0 \), the estimate \( \hat{x}_{{\rm a},i}(t) \), generated by the estimator \eqref{eq:10}, converges exponentially to a neighborhood of \( x_{\rm a}(t) \), that is,  
\begin{align*}
\lim_{t \to \infty} \sup \left| \hat{x}_{{\rm a},i}(t) - x_{\rm a}(t) \right| &\leq \frac{\gamma_{\rm s}}{\beta \lambda_2}, \\
\sup_{\tau \in [t, \infty)} \left\| \left( I_N - \frac{1}{N} \mathbf{1}_N \mathbf{1}_N^{\rm T} \right) \dot{x}(\tau) \right\| &= \gamma_{\rm s} < \infty,
\end{align*}
where \( \lambda_2 \) represents the smallest positive eigenvalue of the Laplacian matrix \( L \).
\end{lemma}

\begin{lemma}\label{lem:2}
\cite{7548310}: There exists a positive constant \( \gamma_{\rm p} > 0 \) such that, for every \( \kappa > 0 \), the estimate \( \hat{p}_{{\rm a},i}(t) \), generated by the estimator \eqref{eq:12}, converges exponentially to a neighborhood of \( p_{\rm a}(t) \), given by
\[
\lim_{t \to \infty} \sup \left| \hat{p}_{{\rm a},i}(t) - p_{\rm a}(t) \right| \leq \frac{\psi \gamma_{\rm p}}{\kappa}.
\]
\end{lemma}

Based on Lemmas \ref{lemma1} and \ref{lem:2}, the following results on the distributed power allocation algorithm \eqref{eq:14} were established in \cite{meng2021distributed}.

\begin{theorem} \cite{meng2021distributed}: Given that the values of \( C_i \) and \( V_i \) are available, and conditions stated in assumptions~\ref{ass: graph}, \ref{ass: power}, and \ref{ass: knowledge} are met, the distributed power allocation law~\eqref{eq:14} successfully solves Problem~1. Specifically, for any predefined accuracy levels \( \epsilon_{\rm s}, \epsilon_{\rm p} > 0 \), there exist sufficiently large parameters \( \beta, \kappa > 0 \) such that, for all \( i \in \mathcal{N} \), both objectives of Problem~1 are achieved.
\end{theorem}

\section{Privacy-Preserving Power Allocation Algorithm Design} \label{sec:privacy}

If eavesdroppers deduce private information via interception, the system’s security could be compromised. Disclosing this information may enable adversaries to launch attacks, increasing electricity generation costs, or even causing a power system outage.

In the networked battery system we are considering in this paper, each battery unit power $p_{i}(t)$ is private information that should not be leaked. Considering equation~\eqref{eq:mode}, we can conclude that $p_i(t)$ can be inferred if the attacker is able to infer $\dot{x}_i(t)$. We know that $x_i(t)$ (for discharging $x_{{\rm d},i} = C_iV_iS_i(t)$, for charging $x_{{\rm c},i} = C_iV_i(1-S_i(t))$) is also private information for each battery unit that should not be leaked. It is assumed that the eavesdropper is aware of the communication network and can observe all information transmitted among the battery units. Therefore, information known to the eavesdropper includes $\hat{x}_{{\rm a},i}$, $\hat{p}_{{\rm a},i}$, $A$, $\beta$, and $\kappa$. 

As discussed in Section \ref{sec:attack}, when dynamic average consensus is achieved using~\eqref{eq:10}, the time-varying internal states \( x_i(t) \) and their derivatives \( \dot{x}_i(t) \) can be inferred by the external eavesdropper through the use of an observer model, assuming that $x_i(t)$, $\dot{x}_i(t)$ and $\ddot{x}_i(t)$ are bounded. Therefore, we cannot guarantee the privacy of the networked battery system.

In the following subsections, we first introduce a privacy-preserving algorithm to estimate the scaled average unit state $\eta x_{{\rm a}}(t)$ of the networked battery system. Next, we discuss how the leader-following consensus algorithm in equation~\eqref{eq:12} is modified to preserve the privacy of the average desired power $p_{\rm a}(t)$. Finally, in the last subsection, these estimators are utilized to develop the power allocation algorithm.

\subsection{Privacy-Preserving Distributed Average Unit State Estimator }

In this section, motivated by the schemes in~\cite{zhang2022privacy, 8657789}, we decompose the distributed estimator state \( \hat{x}_{{\rm a},i} \), the state of the battery unit \( x_i \), and its derivative \( \dot{x}_i \) into two subsets: \( \left\{ \hat{x}_{{\rm a},i}^{\alpha}, x_{i}^{\alpha}, \dot{x}_{i}^{\alpha} \right\} \) and \( \left\{ \hat{x}_{{\rm a},i}^{\beta}, x_{i}^{\beta}, \dot{x}_{i}^{\beta} \right\} \). The initial values are randomly chosen as follows,
\begin{align*}
    x_{i}^{\alpha}(0) + x_{i}^{\beta}(0) = 2\eta x_i(0) \hspace{0pt},
\end{align*}
where \( \eta > 0 \) with $\eta\neq 1$ is a predefined scaling constant introduced as part of the privacy-preserving mechanism. Clearly, $\dot{x}_i^{\alpha}(t)$ and $\dot{x}_i^{\beta}(t)$ are bounded and meet the following constraint,
\begin{align*}
\dot{x}_{i}^{\alpha}(t) + \dot{x}_{i}^{\beta}(t) =2\eta \dot{x}_{i}(t).
\end{align*}

It follows that the sum of the decomposed states evolves consistently with the scaled original state trajectory, as given by
\begin{align*}
x_i^{\alpha}(t) + x_i^{\beta}(t) = 2\eta x_i(t),
\end{align*}
which implies that
\begin{align*}
\frac{1}{2N} \sum_{i=1}^{N} \left( x_i^{\alpha}(t) + x_i^{\beta}(t) \right) = \frac{\eta}{N} \sum_{i=1}^{N} x_i(t) = \eta x_{\rm a}(t).
\end{align*}
Thus, the consensus value converges to a scaled version of $x_{\rm a}(t)$. 
This scaling provides an additional layer of privacy, protecting the privacy of not only the individual states $x_i(t)$ and their derivatives $\dot{x}_i(t)$, but also the average state $x_{\rm a}(t)$.
 In Theorem~\ref{thm:2}, we will prove that sub-states $\hat{x}_{{\rm a},i}^{\alpha}(t)$ and $\hat{x}_{{\rm a},i}^{\beta}(t)$ convergence to a neighborhood of \( \eta x_{\rm a}(t) \).  

Under the decomposition scheme, the sub-state $\hat{x}_{{\rm a},i}^{\alpha}(t)$ serves as a substitute for the original state $\hat{x}_{{\rm a},i}(t)$. This includes responsibility for interaction among battery units and is the only state information of a battery unit $i$ that is communicated to its neighbors $j$. 

The sub-state $\hat{x}_{{\rm a},i}^{\beta}(t)$ is involved in the distributed update process by interacting with the internal sub-state $\hat{x}_{{\rm a},i}^{\alpha}(t)$ of agent $i$. As a result, $\hat{x}_{{\rm a},i}^{\beta}(t)$ influences the evolution of $\hat{x}_{{\rm a},i}^{\alpha}(t)$, while remaining hidden from neighbors of agent $i$. Therefore, only available data to the eavesdropper is $A, \beta$, and  $\hat{x}_{{\rm a},i}^{\alpha}(t) $.

Figure \ref{fig:fig2} presents a visual representation of state-decomposition in a network. After incorporating the decomposition, we reformulate the average unit state estimator \eqref{eq:10} as follows,
\begin{subequations}\label{decomp}
\begin{alignat}{3}
    \dot{\hat{x}}_{{\rm a},i}^{\alpha}(t) &= \dot{x_i}^{\alpha}(t) 
    - \beta \sum_{j=1}^{N} a_{ij} \left({\hat{x}}_{{\rm a},i}^{\alpha}(t) - {\hat{x}}_{{\rm a},j}^{\alpha}(t) \right) \notag \\
    &\quad - \beta \left( {\hat{x}}_{{\rm a},i}^{\alpha}(t) - {\hat{x}}_{{\rm a},i}^{\beta}(t) \right), \label{eq:15} \\
    \dot{\hat{x}}_{{\rm a},i}^{\beta}(t) &= \dot{x_i}^{\beta}(t) - \beta \left( {\hat{x}}_{{\rm a},i}^{\beta}(t) - {\hat{x}}_{{\rm a},i}^{\alpha}(t) \right), \label{eq:16}  \\
    \hat{x}_{{\rm a},i}^{\alpha}(0) &= x_i^{\alpha}(0), \quad \hat{x}_{{\rm a},i}^{\beta}(0) = x_i^{\beta}(0), \quad \ i \in \mathcal{N}. \label{eq:17}
\end{alignat}
\end{subequations}

\subsection{Privacy-Preserving  Distributed Average Desired Power Estimator}

We use \eqref{eq:12} to estimate a scaled version of the average desired power \( p_{\rm a}(t) \), thereby enhancing its privacy. In \eqref{eq:12}, the local estimate \( \hat{p}_{{\rm a},i}(t) \) is shared over the network and thus is accessible to an external eavesdropper. Since \( \hat{p}_{{\rm a},i}(t) \) converges to the true average desired power \( p_{\rm a}(t) \), this exposes sensitive information.

To mitigate this risk, we modify the estimator such that the consensus is performed with respect to a scaled version of the average desired power. Specially, the average desired power \( p_{\rm a}(t) \) is known by at least one designated battery unit, which injects a scaled version \( \sigma p_{\rm a}(t) \) into the network. All battery units know the scaling factor \( \sigma > 0 \) with $\sigma\neq 1$, but the eavesdropper does not. Therefore all other battery units can recover the true average desired power. The modified estimator is given by
\begin{equation} \label{eq:privacy desired power}
\dot{\hat{p}}_{{\rm a},i}(t)\!=\!-\kappa \!\left( \sum_{j=1}^{N} a_{ij} (\hat{p}_{{\rm a},i}(t) \!-\! \hat{p}_{{\rm a},j}(t))\! + b_i (\hat{p}_{{\rm a},i}(t) \!-\! \sigma p_{\rm a}(t))\! \!\right)\!\!.
\end{equation}

Through this scaling approach, the shared local estimates \( \hat{p}_{{\rm a},i}(t) \) converge to \( \sigma p_{\rm a}(t) \) rather than the \( p_{\rm a}(t) \). As a result, even though an eavesdropper can access all \( \hat{p}_{{\rm a},i}(t) \), it cannot infer the actual average desired power without the knowledge of \( \sigma \).

According to Lemma~\ref{lem:2}, 
\[
\lim_{t \to \infty} \sup \left| \hat{p}_{{\rm a},i}(t) - \sigma p_{\rm a}(t) \right| \leq \frac{\psi \gamma_{\rm p}}{\kappa}.
\]
This guarantees exponential convergence of the estimates to a neighborhood of the scaled version of the average desired power while preserving the privacy of the average desired power.
\begin{figure}[!t]
\centering
\includegraphics[width=1\linewidth]{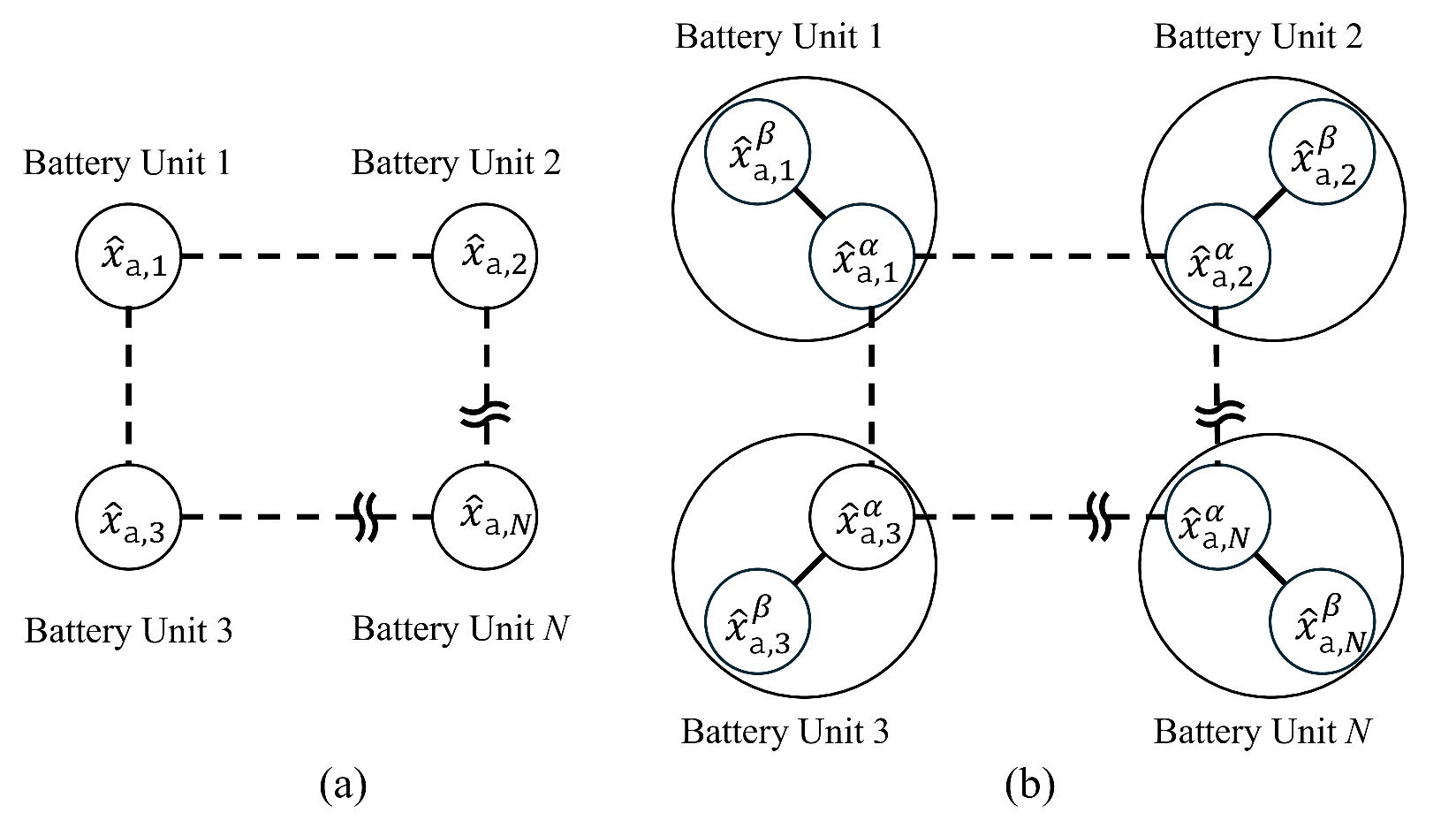}
\caption{Explanation of state-decomposition process: (a) Original state before decomposition. (b) Decomposed state. }
\label{fig:fig2}
\end{figure}

\begin{assumption}
\label{ass:secure_init}
All communication links in the network are encrypted during the initialization phase of the algorithm, enabling the scaling parameters $\eta$ and $\sigma$ to be securely shared among all battery units. 
These parameters are known to all battery units but are inaccessible to any external eavesdropper. Moreover, the scaling parameters $\eta$ and $\sigma$ are reinitialized each time the algorithm is executed and can also be refreshed following network reconfiguration events (e.g., unit addition or removal), thereby preventing long-term inference and further enhancing privacy protection.
\end{assumption}

\subsection{Power Allocation Algorithm Based on Privacy-Preserving Estimators}

With the privacy-preserving estimators, it is ensured that \( \hat{x}_{{\rm a},i}^{\alpha}(t) \) converges to a neighborhood of \( \eta x_{\rm a}(t) \), and \( \hat{p}_{{\rm a},i}(t) \) converges to a neighborhood of \( \sigma p_{\rm a}(t) \). Since the scaling factors are known to all battery units within the system, the original average state \( x_{\rm a}(t) \) and the desired average power \( p_{\rm a}(t) \) can be recovered by dividing the respective estimates by \( \eta \) and \( \sigma \).

Therefore, the power allocation algorithm in \eqref{eq:14} for the battery unit \( i \) is modified to incorporate privacy-preserving estimators as follows,
\begin{equation}
\!\!p_i(t) =
\begin{cases} 
\frac{x_{{\rm d},i}(t)}{\max{\left\{\!\frac{a_1}{2}, \frac{\hat{x}_{{\rm a},i}^{\alpha}(t)}{\eta} \!\right\}}} \frac{\hat{p}_{{\rm a},i}(t)}{\sigma}, & \text{(discharging mode)}, \\
\frac{x_{{\rm c},i}(t)}{\max{\left\{\! \frac{a_1}{2}, \frac{\hat{x}_{{\rm a},i}^{\alpha}(t)}{\eta}\! \right\}}} \frac{\hat{p}_{{\rm a},i}(t)}{\sigma}, & \text{(charging mode)}.
\end{cases}
\label{eq:18}
\end{equation}

\begin{theorem}\label{thm:2}
Consider the proposed decomposition framework in \eqref{eq:15}, \eqref{eq:16}, and \eqref{eq:17}. Let assumption~\ref{ass: graph} hold. The tracking deviations 
\( \hat{x}_{{\rm a},i}^{\alpha}(t) - \frac{\eta}{n} \sum_{j=1}^{N} x_j(t) \) and 
\( \hat{x}_{{\rm a},i}^{\beta}(t) - \frac{\eta}{n} \sum_{j=1}^{N} x_j(t) \) remain ultimately bounded, and the bounds on them can be made arbitrarily small by tuning the value of the design parameter $\beta$.
\end{theorem}

\begin{proof}
The average state estimator given in \eqref{eq:10} can be represented in matrix form as:  
\begin{equation}
    \dot{\hat{x}}_{\rm a}(t) = -\beta L \hat{x}_{\rm a}(t) + \dot{x}(t), \quad \hat{x}_{{\rm a},i} (0) = x_i(0).
    \label{eq:19}
\end{equation}

As shown in \cite{kia2019tutorial}, for a given undirected connected graph, there exists a positive constant \( \gamma_{\rm s} > 0 \) such that, for every \( \beta > 0 \), the estimate \( \hat{x}_{{\rm a},i}(t) \) produced by the estimator in~\eqref{eq:10} converges exponentially to a bounded region around \( x_a(t) \), i.e.,  
\[
\lim_{t \to \infty} \sup \left| \hat{x}_{{\rm a},i}(t) - x_{\rm a}(t) \right| \leq \frac{\gamma_{\rm s}}{\beta \lambda_2},
\]

\[
\sup_{\tau \in [t, \infty)} \left\| \left( I_N - \frac{1}{N} \mathbf{1}_N \mathbf{1}_N^{\rm T} \right) \dot{x}(\tau) \right\| = \gamma_{\rm s} < \infty,
\]
where \( \lambda_2 \) denotes the smallest positive eigenvalue of the Laplacian matrix \( L \), which characterizes the connectivity of the graph.  

Consequently, \eqref{eq:15} and \eqref{eq:16} can be reformulated in matrix form, analogous to \eqref{eq:19}, but incorporating a modified Laplacian matrix to account for the decomposition-induced graph structure: 
\begin{equation}
    \dot{\hat{x}}_{\rm a}^{\text{decomp}} = -\beta L' \hat{x}_{\rm a}^{\text{decomp}} + \dot{x}^{\text{decomp}}, 
\end{equation}
    \label{eq:20}
where 
\begin{align*}
\hat{x}_{\rm a}^{\text{decomp}} &=
\begin{bmatrix}
    \hat{x}_{{\rm a},1}^{\alpha}(t) \; \dots \; \hat{x}_{{\rm a},N}^{\alpha}(t) \; \; \hat{x}_{{\rm a},1}^{\beta}(t) \; \dots \; \hat{x}_{{\rm a},N}^{\beta}(t)
\end{bmatrix}^{\rm T}, \\
x^{\text{decomp}} &=
\begin{bmatrix}
    x_{1}^{\alpha}(t) \; \dots \; x_{N}^{\alpha}(t) \;\; x_{1}^{\beta}(t) \; \dots \; x_{N}^{\beta}(t)
\end{bmatrix}^{\rm T}.
\end{align*}
The modified decomposition-based Laplacian matrix is defined as 
\[
L' =
\begin{bmatrix}
L + I_N & -I_N \\
-I_N & I_N
\end{bmatrix},
\]
where \( I_N \in \mathbb{R}^{N \times N} \) denotes the identity matrix of dimension \( N\). 



    

It is noted that the new Laplacian matrix \( L' \) remains symmetric, as \( L \) corresponds to the Laplacian of an undirected connected graph. Consequently, by Lemma~\ref{lemma1}, dynamic average consensus can be achieved.

The convergence error is given by
\begin{equation}\label{eq:21}
\lim_{t \to \infty} \sup \left| \hat{x}_{{\rm a},i}^{\alpha}(t) - 
\frac{1}{2N} \sum_{i=1}^{N} \left( x_i^{\alpha}(t) + x_i^{\beta}(t) \right) 
\right| \leq \frac{\gamma_{\rm s}'}{\beta {\lambda}_2^{'}},
 \end{equation}
\begin{equation}\label{eq:22}
\lim_{t \to \infty} \sup \left| \hat{x}_{{\rm a},i}^{\beta}(t) - 
\frac{1}{2N} \sum_{i=1}^{N} \left( x_i^{\alpha}(t) + x_i^{\beta}(t) \right) 
\right| \leq \frac{\gamma_{\rm s}'}{\beta {\lambda}_2^{'}},
\end{equation}
\begin{equation}
\sup_{\tau \in [t, \infty)} \left\| \left( I_{2N} - \frac{1}{2N} \mathbf{1}_{2N} \mathbf{1}_{2N}^{\rm T} \right) \dot{x}^{\rm decomp}(\tau) \right\| = \gamma_{\rm s}' < \infty,
\end{equation}
where $\lambda_2^{'}$ represents the smallest positive eigenvalue of the Laplacian matrix $L'$ and $\gamma_{\rm s}'>0$ is a scalar constant.

Given the initial conditions,
\[
x_i^{\alpha}(0) + x_i^{\beta}(0) = 2\eta x_i(0),
\]
and considering that \( \dot{x}_i^{\alpha}(t) \) and \(  \dot{x}_i^{\beta}(t) \) are bounded and selected to satisfy
\[
\dot{x}_i^{\alpha}(t) + \dot{x}_i^{\beta}(t) = 2\eta\dot{x}_i(t),
\]
we have
\[
x_i^{\alpha}(t) + x_i^{\beta}(t) = 2\eta x_i(t),
\]
which implies that
\[
\frac{1}{2N} \sum_{i=1}^{N} \left( x_i^{\alpha}(t) + x_i^{\beta}(t) \right) = \frac{\eta}{N} \sum_{i=1}^{N} x_i(t) = \eta x_{\rm a}(t).
\]

Thus, \eqref{eq:21} and \eqref{eq:22} can be simplified as
\[
\lim_{t \to \infty} \sup \left| \hat{x}_{{\rm a},i}^{\alpha}(t) - 
 \eta x_{\rm a}(t)
\right| \leq \frac{\gamma_{\rm s}'}{\beta {\lambda}_2^{'}}, 
\]
\[
\lim_{t \to \infty} \sup \left| \hat{x}_{{\rm a},i}^{\beta}(t) - 
 \eta x_{\rm a}(t) 
\right| \leq \frac{\gamma_{\rm s}'}{\beta {\lambda}_2^{'}}.
\]

Therefore, all sub-states \( \hat{x}^{\alpha}_{{\rm a},i}(t) \) and \( \hat{x}^{\beta}_{{\rm a},i}(t) \) in \eqref{eq:15} and \eqref{eq:16} converge to a neighborhood of a scaled version of the average consensus value corresponding to the original states, and the bounds on the steady-state errors can be made arbitrarily small by tuning the value of $\beta$.
\end{proof}

\begin{remark}\label{remark1}
The second smallest eigenvalue of the Laplacian matrix of a graph characterizes the convergence rate of consensus algorithms \cite{4118472}. Since introducing the parameter $\eta$ in the state-decomposition mechanism does not alter the Laplacian, the convergence rate remains identical to that of the conventional decomposition case. Furthermore, by Theorem~\ref{thm:2}, it is straightforward to see that the tracking errors are ultimately bounded in the same manner as in the conventional state-decomposition case.
\end{remark}

\begin{theorem}
Let the BESS be composed of \( N \) interconnected battery units with known values of \( C_i \) and \( V_i \), and assume that the conditions specified in Assumptions~\ref{ass: graph}, \ref{ass: power}, and~\ref{ass: knowledge} are satisfied. The distributed power allocation law~\eqref{eq:18}, combined with the privacy-preserving estimators~\eqref{eq:15},~\eqref{eq:16}, and~\eqref{eq:privacy desired power}, solves Problem~1. Specifically, for any predefined accuracy levels \( \epsilon_{\rm s}, \epsilon_{\rm p} > 0 \), there exist sufficiently large parameters \( \beta, \kappa > 0 \) such that, for all \( i \in \mathcal{N} \), both objectives of Problem~1 are achieved.

\end{theorem}

\begin{proof} 

We analyze the discharge phase and note that the charging phase follows similarly.

We define the estimation errors as
\begin{align*}
    e_{{x},i}(t) &= \frac{\hat{x}_{{\rm a},i}^\alpha(t)}{\eta} - x_{\rm a}(t), \quad e_{{p},i}(t) =\frac{\hat{p}_{{\rm a},i}(t)}{\sigma} - p_{\rm a}(t).
\end{align*}

Selecting \( \beta \geq \frac{2\gamma_{\rm s}'}{a_1 \lambda_2'} \), where \( \gamma_{\rm s}' \) is defined in Theorem~\ref{thm:2}, we use Theorem~\ref{thm:2} to conclude
\begin{equation*}
    \lim_{t \to \infty} \sup |e_{{x},i}(t)| < \frac{a_1}{2\eta}.
\end{equation*}
Since \( x_{\rm a}(t) \geq a_1 \) for all \( t \geq 0 \), at steady state, 
\begin{equation*}
    \hat{x}_{{\rm a},i}^\alpha(t) \geq a_1\eta - \frac{a_1}{2}.
\end{equation*}

By Lemma~\ref{lem:2}, \( e_{{\rm p},i}(t) \) converges exponentially to zero, leading to
\begin{equation*}
    \lim_{t \to \infty} \sup |e_{p,i}(t)| \leq \frac{\psi \gamma_{\rm p}}{\sigma\kappa}.
\end{equation*}
Thus, the relative error in \( p_{\rm a} \) satisfies
\begin{equation*}
    \lim_{t \to \infty} \sup \frac{|e_{{p},i}(t)|}{p_{\rm a}(t)} \leq \frac{N \psi \gamma_{\rm p}}{\sigma\kappa p}.
\end{equation*}

The steady-state discharging power corresponding to the \( i \)th unit is
\begin{equation*}
    p_i(t) = \frac{x_{{\rm d},i}(t)}{x_{\rm a}(t) + e_{{x},i}(t)} (p_{\rm a}(t) + e_{{p},i}(t)).
\end{equation*}
Rearranging, we obtain
\begin{equation*}
    p_i(t) = \frac{\left(1 + \frac{e_{{p},i}(t)}{p_{\rm a}(t)}\right) x_{{\rm d},i}(t)}{\left(1 + \frac{e_{{x},i}(t)}{x_{\rm a}(t)}\right)x_{\rm a}(t)} p_{\rm a}(t).
\end{equation*}
In steady state, it follows that
\begin{equation*}
    \frac{1 - \frac{N \psi \gamma_{\rm p}}{\sigma\kappa p}}{1 + \frac{\gamma_{\rm s}'}{a_1 \eta\beta \lambda_2'}} \leq \frac{1 + \frac{e_{{p},i}(t)}{p_{\rm a}(t)}}{1 + \frac{e_{{x},i}(t)}{x_{\rm a}(t)}} \leq   \frac{1 + \frac{N \psi \gamma_{\rm p}}{\sigma\kappa p}}{1 - \frac{\gamma_{\rm s}'}{a_1\eta \beta \lambda_2'}}
\end{equation*}

We define a new deviation term
\begin{equation*}
    \xi_i(t) = \frac{1 + \frac{e_{{p},i}(t)}{p_{\rm a}(t)}}{1 + \frac{e_{{x},i}(t)}{x_{\rm a}(t)}} - 1.
\end{equation*}
Bounding \( \xi_i \),
\begin{align*}
    \xi^- &= \frac{1 - \frac{N \psi \gamma_{\rm p}}{\sigma\kappa p}}{1 + \frac{\gamma_{\rm s}'}{a_1\eta\beta \lambda_2'}} - 1, \quad 
    \xi^+ = \frac{1 + \frac{N \psi \gamma_{\rm p}}{\sigma\kappa p}}{1 - \frac{\gamma_{\rm s}'}{a_1\eta\beta \lambda_2'}} - 1
\end{align*}
Thus, in steady state,
\begin{equation*}
    \xi^- \leq \xi_i(t) \leq \xi^+.
\end{equation*}
Choosing sufficiently large $\beta$ and $\kappa$ minimizes these bounds.

Define
\begin{equation*}
    r(t) = \frac{p_{{\rm a}}(t)}{x_{\rm a}(t)}.
\end{equation*}
Expressing $p_i(t)$ as a function of $r(t)$ and $\xi_i(t)$, we have
\begin{equation*}
    p_i(t) = r (1 + \xi_i(t)) {x}_{{\rm d},i}(t).
\end{equation*}
Since $r(t) \geq \frac{p}{Na_2} > 0$ since  $x_{\rm a}(t) \leq a_2$. Also, $\dot{s}_i(t) =- \frac{1}{C_i V_i} p_i(t)$ and $x_{{\rm d},i}(t)=C_iV_is_i$(t), it follows

\begin{equation*}
    \dot{s}_i(t) = -r(t)(1 + \xi_i(t)) s_i(t).
\end{equation*}
Consider the function
\begin{equation*}
    W_{ij} = \frac{1}{2} (s_i(t) - s_j(t))^2.
\end{equation*}
Its derivative is
\begin{equation*}
    \dot{W}_{ij} = -r(t)(s_i(t) - s_j(t))((1 + \xi_i(t))s_i(t) - (1 + \xi_j(t))s_j(t)).
\end{equation*}
We observe that $\dot{W}_{ij} < 0$ if
\begin{equation*}
    (s_i(t) - s_j(t))((1 + \xi_i(t))s_i(t) - (1 + \xi_j(t))s_j(t)) > 0.
\end{equation*}
This condition holds when
\[
  \begin{aligned}
    & s_i(t) > s_j(t) \quad \text{and} \quad  \frac{s_i(t)}{s_j(t)} > \frac{1 + \xi_j(t)}{1 + \xi_i(t)}, \\
    & \quad \text{or} \quad s_i(t)< s_j(t) \quad \text{and} \quad  \frac{s_i(t)}{s_j(t)} < \frac{1 + \xi_j(t)}{1 + \xi_i(t)}.
  \end{aligned}
\]
Thus, in steady state, \begin{equation*}
\xi^{-} \leq \xi_i(t) \leq \xi^{+}.
\end{equation*}
\begin{equation*}
    \frac{1 + \xi^-}{1 + \xi^+} \leq \frac{s_i(t)}{s_j(t)} \leq \frac{1 + \xi^+}{1 + \xi^-}.
\end{equation*}
Ensuring convergence to the set
\begin{equation*}
    \left\{(s_i, s_j) : |s_i(t) - s_j(t)| \leq \frac{\xi^+ - \xi^-}{1 + \xi^-} \right\}.
\end{equation*}
Choosing $\beta, \kappa$ such that
\begin{equation*}
    \frac{\xi^+ - \xi^-}{1 + \xi^-} \leq \varepsilon_{\rm s},
\end{equation*}
ensures SoC balancing among battery units with any predefined level of accuracy. The total power output is determined by
\begin{equation*}
    p_{\Sigma}(t) = \sum_{i=1}^{N} \frac{(1 + \frac{e_{{p},i}(t)}{p_{\rm a}(t)}) x_{{\rm d},i}(t)}{(1+\frac{e_{{x},i}(t)}{x_{\rm a}(t)})x_{\rm a}(t)} p_{\rm a}(t).
\end{equation*}
Bounding the output power
\begin{equation*}
    (1 + \xi^-) p^*(t) \leq p_{\Sigma}(t) \leq (1 + \xi^+) p^*(t).
\end{equation*}
Choosing $\beta, \kappa$ such that
\begin{equation*}
    \xi^+ \leq \varepsilon_{\rm p},
\end{equation*}
ensures the delivery of power with any predefined accuracy while achieving SoC balancing among battery units with any predefined level of accuracy. Therefore, control objectives are achieved.
\end{proof}

\begin{theorem}\label{thm:4}
Under the modified state-decomposition mechanism in \eqref{decomp}, 
an external eavesdropper cannot infer the internal state $x_p(t)$ 
\emph{and its derivative} $\dot x_p(t)$ of any battery unit $p$ with guaranteed accuracy.
\end{theorem}

\begin{proof}
Under the decomposition scheme, each unit signal $\{x_i(t),\dot x_i(t)\}$ is divided into
two sub-signals $\{x_i^\alpha(t),\dot x_i^\alpha(t)\}$ and $\{x_i^\beta(t),\dot x_i^\beta(t)\}$.
Moreover, under the estimator \eqref{decomp}, 
only $\hat x^\alpha_{{\rm a},i}(t)$ is exchanged.
Therefore, the information accessible to an external eavesdropper is
\[
I_{\mathrm{dec}}(t)\triangleq\{A,\beta,\hat x^\alpha_{{\rm a},i}(t)\}_{i\in\mathcal N}.
\]
Let
$\{\bar x_i(t),\dot{\bar x}_i(t),\bar x_i^\alpha(t),\dot{\bar x}_i^\alpha(t),
\bar x_i^\beta(t),\dot{\bar x}_i^\beta(t)\}$ 
and
$\{\hat{\bar x}^\alpha_{{\rm a},i}(t),\hat{\bar x}^\beta_{{\rm a},i}(t)\}$
be another set of signals and their corresponding estimator states, respectively.
To show that the privacy of $x_p(t)$ and $\dot x_p(t)$ can be protected against the eavesdropper,
it suffices to show that under any bounded perturbations
\[
\epsilon_p(t)=\epsilon_p(0)+\int_0^t\delta_p(\tau)\,{\rm d}\tau
\]
that alter the private signals to
\[
\bar x_p(t)=x_p(t)+\epsilon_p(t),\qquad \dot{\bar x}_p(t)=\dot x_p(t)+\delta_p(t),
\]
there exists a new constructed set of estimator variables such that
\[
\bar I_{\mathrm{dec}}(t)\triangleq\{A,\beta,\hat{\bar x}^\alpha_{{\rm a},i}(t)\}_{i\in\mathcal N}
\]
is exactly the same as $I_{\mathrm{dec}}(t)$.
This is because the only information available to the eavesdropper is $I_{\mathrm{dec}}(t)$, and if
$I_{\mathrm{dec}}(t)$ could be the outcome of any perturbed values of $x_p(t)$ and $\dot x_p(t)$,
then the eavesdropper has no way to even find a range for $x_p(t)$ and $\dot x_p(t)$.

Construct another set of private signals with the same average.
Pick any $l\in\mathcal N\setminus\{p\}$.
Define
\begin{align}
\bar x_p(t) &= x_p(t)+\epsilon_p(t), &
\dot{\bar x}_p(t) &= \dot x_p(t)+\delta_p(t), \label{eq:bar_p}\\
\bar x_l(t) &= x_l(t)-\epsilon_p(t), &
\dot{\bar x}_l(t) &= \dot x_l(t)-\delta_p(t), \label{eq:bar_l}\\
\bar x_q(t) &= x_q(t), &
\dot{\bar x}_q(t) &= \dot x_q(t), \; \forall q\in\mathcal N\setminus\{p,l\}. \label{eq:bar_others}
\end{align}
Then, $\sum_{j=1}^N\bar x_j(t)=\sum_{j=1}^N x_j(t)$ for all $t$, so the average signal is unchanged.

Let the new set of private signals satisfy the same constraints as the original scheme,
\begin{equation}
\bar x_i^\alpha(0)+\bar x_i^\beta(0)=2\eta\,\bar x_i(0),\;
\dot{\bar x}_i^\alpha(t)+\dot{\bar x}_i^\beta(t)=2\eta\,\dot{\bar x}_i(t),\; \forall i\in\mathcal N.
\label{eq:bar_decomp_constraints}
\end{equation}

Consider the estimator \eqref{decomp} 
driven by new set of decomposed derivatives,
\begin{align}
\dot{\hat{\bar x}}^\alpha_{\rm a,i}(t)
&= \dot{\bar x}^\alpha_i(t)
-\beta\sum_{j=1}^N a_{ij}\big(\hat{\bar x}^\alpha_{{\rm a},i}(t)-\hat{\bar x}^\alpha_{{\rm a},j}(t)\big)
-\beta\big(\hat{\bar x}^\alpha_{{\rm a},i}(t)-\hat{\bar x}^\beta_{{\rm a},i}(t)\big), \label{eq:bar_est_alpha}\\
\dot{\hat{\bar x}}^\beta_{{\rm a},i}(t)
&= \dot{\bar x}^\beta_i(t)
-\beta\big(\hat{\bar x}^\beta_{{\rm a},i}(t)-\hat{\bar x}^\alpha_{{\rm a},i}(t)\big), \label{eq:bar_est_beta}\\
\hat{\bar x}^\alpha_{{\rm a},i}(0)&=\bar x_i^\alpha(0),\;
\hat{\bar x}^\beta_{{\rm a},i}(0)=\bar x_i^\beta(0),\; i\in\mathcal N. \label{eq:bar_est_init}
\end{align}
Choose the initial conditions as follows,
\[
\bar x_i^\alpha(0)=x_i^\alpha(0),\; \forall i\in\mathcal N,
\]
and using the decomposition constraint \eqref{eq:bar_decomp_constraints}
under the perturbed signals defined in \eqref{eq:bar_p}--\eqref{eq:bar_others} results in
\begin{align}
\bar x_p^\beta(0) &= x_p^\beta(0) + 2\eta\,\epsilon_p(0), \label{eq:bar_init_p}\\
\bar x_l^\beta(0) &= x_l^\beta(0) - 2\eta\,\epsilon_p(0), \label{eq:bar_init_l}\\
\bar x_q^\beta(0) &= x_q^\beta(0), \; \forall q\in\mathcal N\setminus\{p,l\}. \label{eq:bar_init_others}
\end{align}

Construct derivatives of the new decomposed signal that keep the $\alpha$ channel unchanged.
Define
\begin{align}
\dot{\bar x}_i^\alpha(t) &= \dot x_i^\alpha(t), \; \forall i\in\mathcal N, \label{eq:bar_xalpha_dot_all}\\
\dot{\bar x}_p^\beta(t) &= \dot x_p^\beta(t)+2\eta\,\delta_p(t), \label{eq:bar_xbeta_dot_p}\\
\dot{\bar x}_l^\beta(t) &= \dot x_l^\beta(t)-2\eta\,\delta_p(t), \label{eq:bar_xbeta_dot_l}\\
\dot{\bar x}_q^\beta(t) &= \dot x_q^\beta(t),\; \forall q\in\mathcal N\setminus\{p,l\}. \label{eq:bar_xbeta_dot_others}
\end{align}
Then, for $i=p$,
\[
\dot{\bar x}_p^\alpha(t)+\dot{\bar x}_p^\beta(t)
= \dot x_p^\alpha(t)+\dot x_p^\beta(t)+2\eta\delta_p(t)
= 2\eta \dot x_p(t)+2\eta\delta_p(t)
= 2\eta \dot{\bar x}_p(t),
\]
and similarly for $i=l$ and $i=q\notin\{p,l\}$, so \eqref{eq:bar_decomp_constraints} holds.

We claim that the following is a solution of \eqref{eq:bar_est_alpha}--\eqref{eq:bar_est_init},
\begin{align}
\hat{\bar x}^\alpha_{{\rm a},i}(t) &= \hat x^\alpha_{{\rm a},i}(t), \; \forall i\in\mathcal N, \label{eq:bar_sol_alpha}\\
\hat{\bar x}^\beta_{{\rm a},p}(t) &= \hat x^\beta_{{\rm a},p}(t)+2\eta\,\epsilon_p(t), \label{eq:bar_sol_beta_p}\\
\hat{\bar x}^\beta_{{\rm a},l}(t) &= \hat x^\beta_{{\rm a},l}(t)-2\eta\,\epsilon_p(t), \label{eq:bar_sol_beta_l}\\
\hat{\bar x}^\beta_{{\rm a},q}(t) &= \hat x^\beta_{{\rm a},q}(t),\; \forall q\in\mathcal N\setminus\{p,l\}. \label{eq:bar_sol_beta_others}
\end{align}
Indeed, differentiating \eqref{eq:bar_sol_beta_p}--\eqref{eq:bar_sol_beta_l} and using
$\dot\epsilon_p(t)=\delta_p(t)$ gives
\[
\dot{\hat{{\bar x}}}^\beta_{{\rm a},p}(t)=\dot{\hat x}^\beta_{{\rm a},p}(t)+2\eta\delta_p(t),
\;
\dot{\hat{\bar x}}^\beta_{{\rm a},l}(t)=\dot{\hat x}^\beta_{{\rm a},l}(t)-2\eta\delta_p(t).
\]
Substituting \eqref{eq:bar_xalpha_dot_all}--\eqref{eq:bar_xbeta_dot_others} and
\eqref{eq:bar_sol_alpha}--\eqref{eq:bar_sol_beta_others} into the dynamics
\eqref{eq:bar_est_alpha}--\eqref{eq:bar_est_beta}, we verify that all equations are satisfied
(and the initial conditions follow from \eqref{eq:bar_init_p}--\eqref{eq:bar_init_others}).

Therefore, the communicated sub-state satisfies
\[
\hat{\bar x}^\alpha_{{\rm a},i}(t)=\hat x^\alpha_{{\rm a},i}(t),\; \forall i\in\mathcal N,\ \forall t\ge 0,
\]
which implies $\bar I_{\mathrm{dec}}(t)=I_{\mathrm{dec}}(t)$.
Since $\epsilon_p(\cdot)$ and $\delta_p(\cdot)$ can be any bounded signals,
the eavesdropper cannot infer $x_p(t)$ and $\dot x_p(t)$ with guaranteed accuracy.
\end{proof}

\begin{remark}\label{remark2}
Under the proposed privacy-preserving control laws, an external eavesdropper with the knowledge of the communication network and all transmitted information (i.e., \( \hat{x}_{{\rm a},i} \), \( \hat{p}_{{\rm a},i} \), \( A \), \( \beta \), and \( \kappa \)) cannot accurately infer the state \( x_i(t) \) or its derivative \( \dot{x}_i(t) \) for any battery unit \( i \), or the average unit state \( x_{\rm a}(t) \). As a result, the eavesdropper is also unable to determine either the power delivery \( p_i(t) \) or its average \( p_{\rm a}(t) \) with guaranteed accuracy.

 Although the decomposition mechanism in~\cite{zhang2022privacy} prevents the reconstruction of \( x_i(t) \) and \( \dot{x}_i(t) \), the average unit state \( x_{\rm a}(t) \) can still be inferred. Motivated by this observation, the modified decomposition mechanism proposed in this paper enhances privacy by ensuring that the sub-state \( \hat{x}^{\alpha}_{{\rm a},i}(t) \) in~\eqref{eq:15} converges toward \( \eta x_{\rm a}(t) \), instead of $x_{\rm a}(t)$, thereby preventing inference of \( x_{\rm a}(t) \) by the eavesdropper in the absence of the knowledge of $\eta$. Similarly, the average desired power estimator in~\eqref{eq:privacy desired power} converges toward \( \sigma p_{\rm a}(t) \). Thus, without the knowledge of $\sigma$, the eavesdropper cannot infer $p_{\rm a}(t)$ either.

\end{remark}

\subsection{Robustness of Privacy Guarantees}

Under Assumption~\ref{ass:secure_init}, the external eavesdropper has no access to the scaling parameters $\eta$ and $\sigma$. We also consider rare scenarios in which partial information leakage may occur. If an external eavesdropper were able to obtain the scaling parameter $\eta$, the proposed privacy-preserving dynamic average consensus algorithm in \eqref{decomp} would reduce to the conventional state-decomposition-based dynamic average consensus algorithm. Even in this case, the privacy of the internal states $x_i(t)$ and $\dot{x}_i(t)$, and hence the individual power outputs $p_i(t)$, would remain protected against external eavesdropper attacks, as discussed in~\cite{zhang2022privacy}. However, knowledge of $\eta$ would allow the eavesdropper to infer the network average state $x_{\rm a}(t)$, which is further discussed in Section~5.3 for the special case $\eta = 1$.
Similarly, if the external eavesdropper were to obtain the scaling parameter $\sigma$, then according to~\eqref{eq:privacy desired power}, the average desired power $p_{\rm a}(t)$ could be inferred. Importantly, even under such partial information leakage, any privacy breach is limited to aggregate quantities such as $x_{\rm a}(t)$ or $p_{\rm a}(t)$, while the privacy of individual internal states remains preserved.

\subsection{Implementation Considerations and Practical Feasibility}

We discuss the computational and communication overhead associated with the proposed privacy-preserving mechanisms. For the privacy-preserving dynamic average consensus algorithm \eqref{decomp}, a modified state-decomposition approach is employed, which differs from the method in Reference~[18] only by the inclusion of an additional scaling factor $\eta$. This modification results in only a negligible increase in local computation, since the additional operations introduced by $\eta$ are simple scalar scalings with $\mathcal{O}(1)$ complexity per update. The overall per-agent computational complexity remains dominated by  neighbor-to-neighbor terms and therefore scales linearly with the number of neighboring units, i.e., $\mathcal{O}(|\mathcal N_i|)$.

 The communication overhead remains unchanged compared to the algorithm in (1) and state-decomposition algorithm in reference [18] since the same amount of information is exchanged among neighboring agents. Consequently, this modified state-decomposition method has a low computational overhead compared with cryptography-based privacy-preserving average consensus schemes \cite{ruan2019secure,chen2022privacy}, which incur significant computational and communication overhead due to repeated encryption and decryption operations \cite{chen2024new}.

For the privacy-preserving leader-following algorithm \eqref{eq:privacy desired power} that is used to estimate the average desired power, privacy is achieved through a simple scaling of the leader’s shared signal. This operation requires only scalar multiplications at both the leader and follower agents and therefore introduces only a negligible additional computational overhead. That is, the scaling operation introduces an additional computational cost of one scalar multiplication per agent per update, corresponding to $\mathcal{O}(1)$ complexity, while the communication overhead remains unchanged.

 As stated in Assumption~\ref{ass:secure_init}, the scaling parameters $\eta$ and $\sigma$ are securely distributed using encryption only during the initialization phase, and no cryptographic processing is required during normal operation. The proposed framework relies solely on lightweight arithmetic operations during runtime and is well suited for real-time implementation.

 Measurement noise in the local state $x_i(t)$ may affect the computed power command $p_i(t)$ through the power allocation law in~(18), since $p_i(t)$ depends on $x_i(t)$. However, the proposed algorithm does not compute the state derivative $\dot{x}_i(t)$ by numerically differentiating noisy measurements. Instead, the state derivative is obtained directly from the system model as
$
\dot{x}_i(t) = -p_i(t)\text{ for discharging}$ and 
$\dot{x}_i(t) = p_i(t)\text{ for charging}$.
Therefore, measurement noise enters the closed-loop system through $x_i(t)$ but is not amplified by numerical differentiation. Consequently, bounded measurement noise leads to bounded effects on system behavior. The proposed algorithm does not explicitly incorporate noise filtering. A detailed analysis and experimental evaluation under high measurement noise is left for future work.

\section{Simulation Results}\label{sec:simulation}

This subsection provides a performance evaluation of our privacy-protected algorithm for controlling a networked BESS. The simulation was performed on a BESS comprising a network of six battery units.  

The battery unit parameters are assigned as
\begin{flalign*}
    & (C_1, C_2, C_3, C_4, C_5, C_6) = (180, 190, 200, 210, 220, 230)~\text{Ah}, &\\
    & (V_1, V_2, V_3, V_4, V_5, V_6) = (50, 50, 50, 50, 50, 50)~\text{V}. &
\end{flalign*}

The communication topology is modeled as an undirected and connected graph, as illustrated in Figure~\ref{fig:fig3}. Let the desired total power \( p^{*}(t) \) be available only to battery unit 1. This implies that \( b_1 = 1 \) and \( b_i = 0 \) for \( i = 2, 3, 4, 5, 6 \).
\begin{figure}[!t]
\centering
\includegraphics[width=0.4\linewidth]{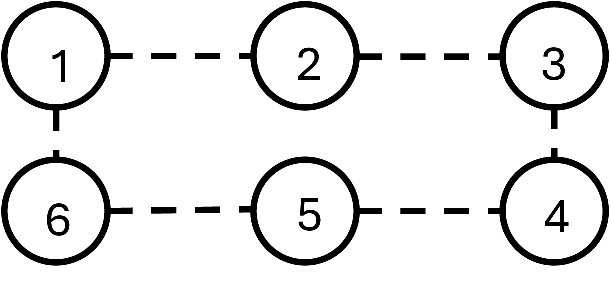}
\caption{The communication topology of the BESS.}
\label{fig:fig3}
\end{figure}
The privacy-preserving distributed control law for power allocation is given in~\eqref{eq:18} with the corresponding estimators defined in~\eqref{eq:15},~\eqref{eq:16}, and~\eqref{eq:privacy desired power}. We set the initial conditions of the estimators as \( \hat{x}_{{\rm a},i}^{\alpha}(0) = x_i^{\alpha}(0) \) and \( \hat{x}_{{\rm a},i}^{\beta}(0) = x_i^{\beta}(0) \), such that \( x_i^{\alpha}(0) \) and \( x_i^{\beta}(0) \) are randomly chosen to satisfy \( x_i^{\alpha}(0) + x_i^{\beta}(0) = 2\eta x_i(0) \). In addition, $\hat{p}_{a,i}(0) = 0$. In the privacy-preserving control framework, the design parameters and predefined scaling factors are specified as $\beta = 300$, $\kappa = 210$, $\eta = 3$, and $\sigma = 4$.

\subsection{Discharging Mode Performance}

In the discharge mode, the desired total power is specified as \( p^{*}(t) = 4200\sin(t) + 4200 \)~W, and the initial SoC values for the six battery units are \( (0.96, 0.89, 0.75, 0.8, 0.73, 0.88) \).

\begin{figure}[!t]
\centering
\includegraphics[width=4.8in, height=1.6in]{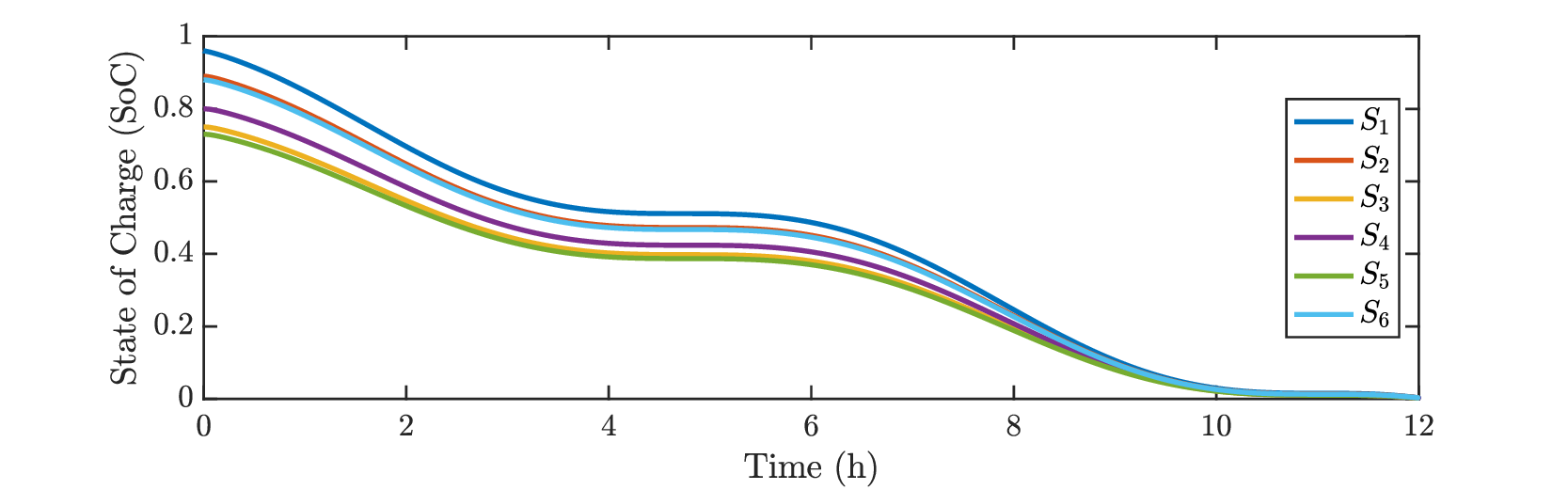}
\caption{The state of charge of each unit over time.}
\label{fig:fig4}
\end{figure}

\begin{figure}[!t]
\centering
\includegraphics[width=4.8in, height=1.6in]{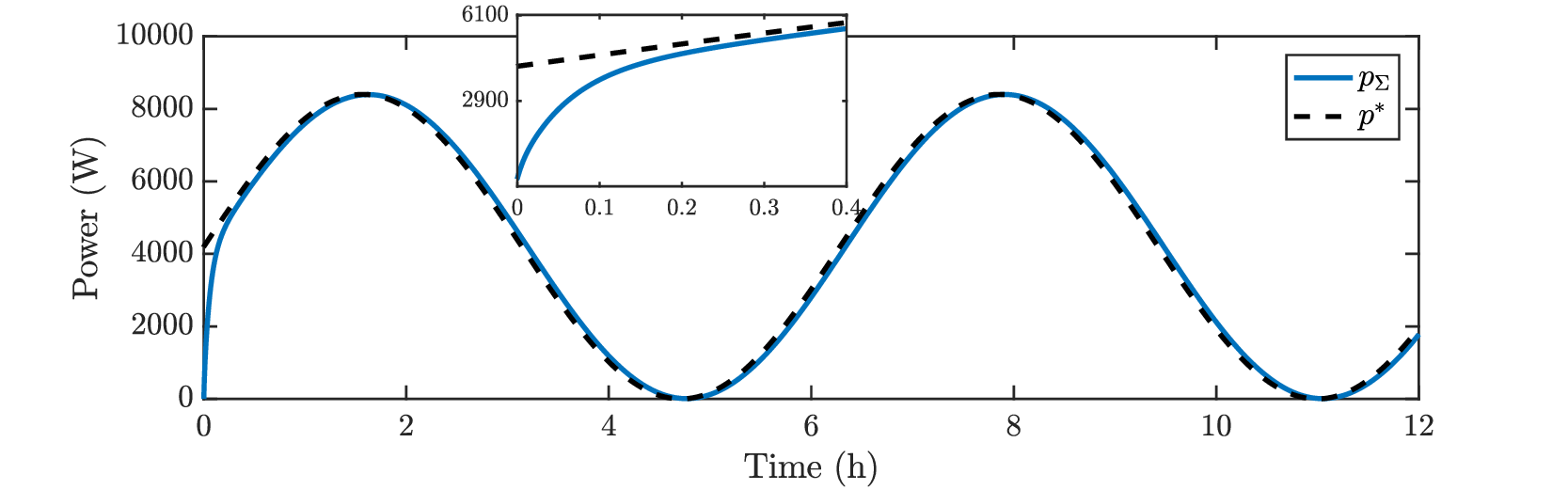}
\caption{The desired power and the total discharging power.}
\label{fig:fig5}
\end{figure}

\begin{figure}[!t]
\centering
\includegraphics[width=4.8in, height=1.6in]{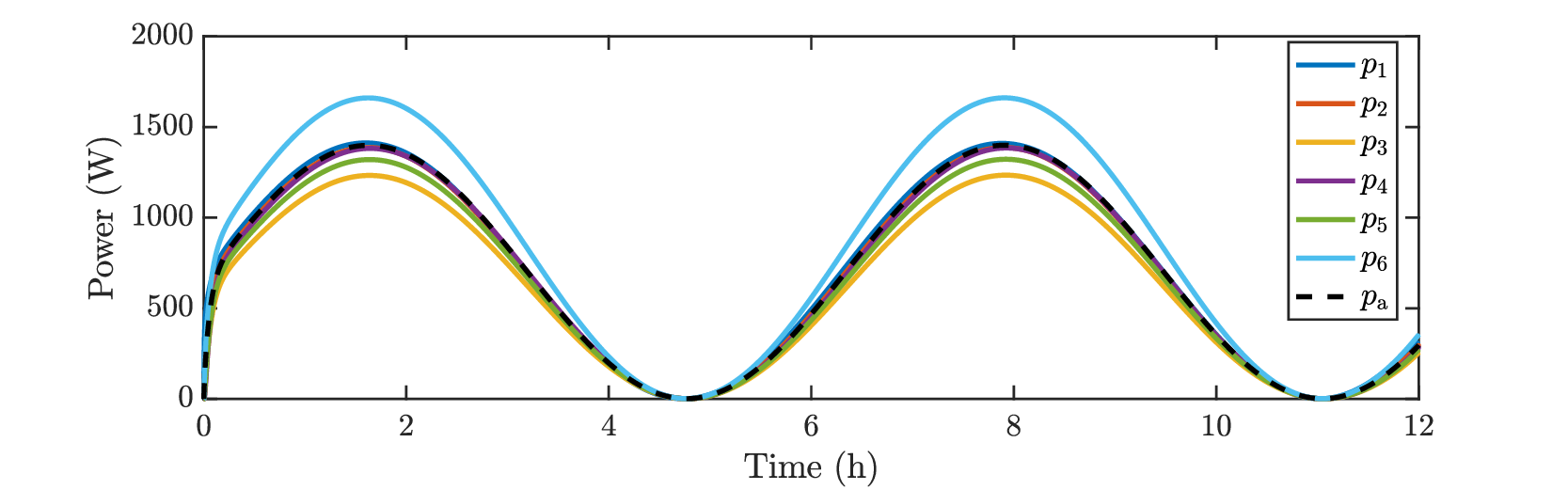}
\caption{The output power delivered by each battery.}
\label{fig:fig6}
\end{figure}

\begin{figure}[!t]
\centering
\includegraphics[width=4.8in, height=1.6in]{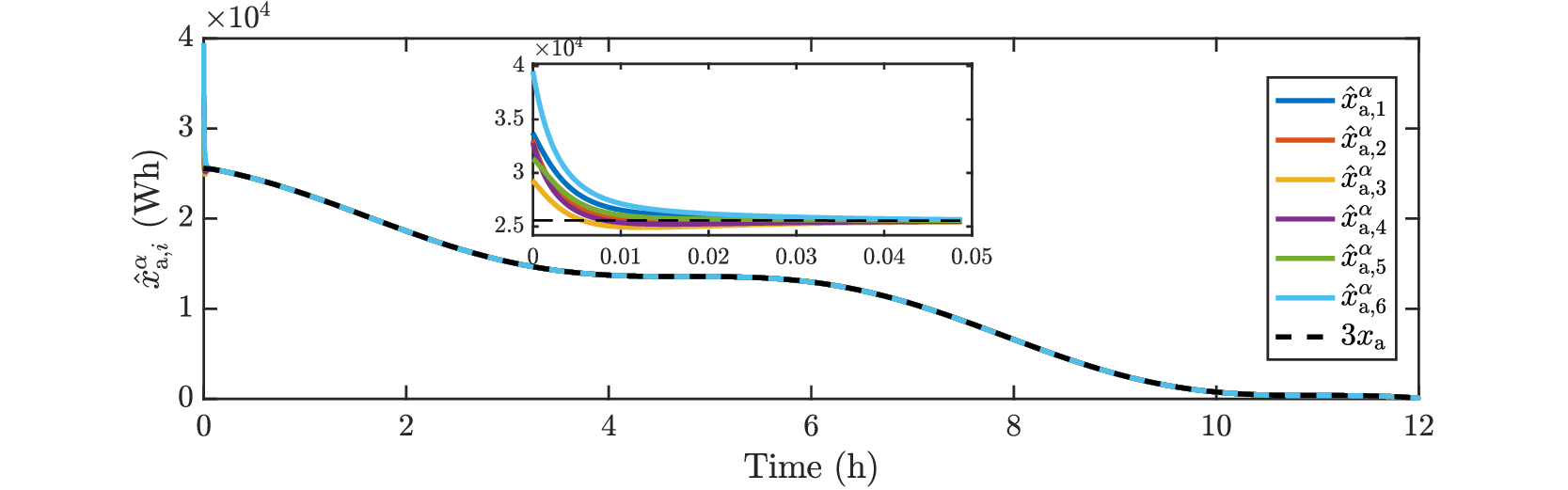}
\caption{The sub-state $\hat{x}_{{\rm a},i}^{\alpha}$ convergence to the scaled average unit state $\eta x_{\rm a}$.}
\label{fig:fig7}
\end{figure}

\begin{figure}[!t]
\centering
\includegraphics[width=4.8in, height=1.6in]{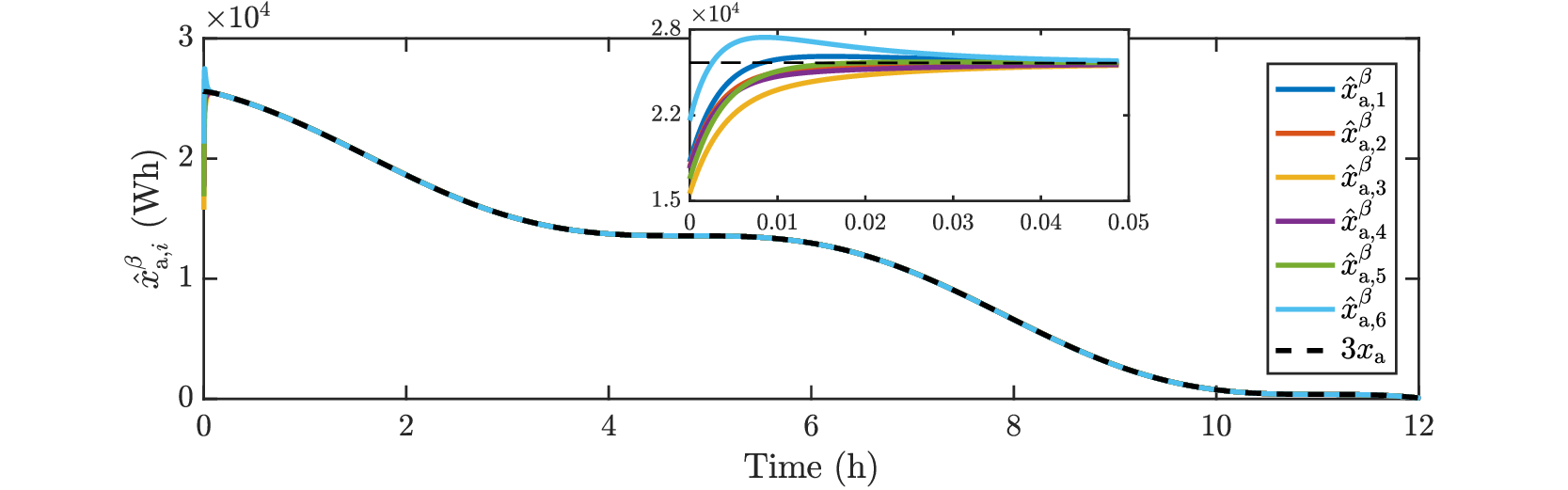}
\caption{The sub-state $\hat{x}_{{\rm a},i}^{\beta}$ convergence to the scaled average unit state $\eta x_{\rm a}$.}
\label{fig:fig8}
\end{figure}

\begin{figure}[!t]
\centering
\includegraphics[width=4.8in, height=1.6in]{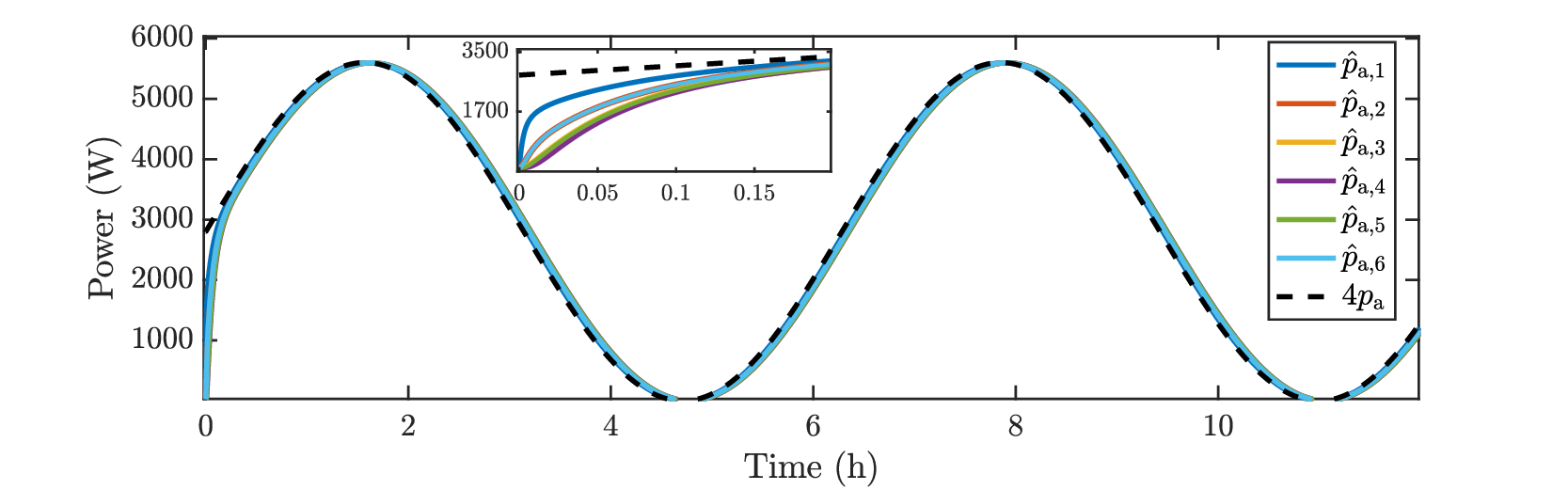}
\caption{The estimation of the scaled average desired power $\sigma p_{\rm a}$.}
\label{fig:fig9}
\end{figure}

Note that a positive value of \( p_{\Sigma}(t) \) signifies discharging operation of the BESS, whereas a negative value corresponds to charging. Figures~\ref{fig:fig4} through~\ref{fig:fig9} show the results of simulations in the discharging mode using the proposed distributed power allocation law with privacy-preserving estimators. The key results are summarized as follows.
\begin{itemize}
    \item Figure~\ref{fig:fig4} shows that, starting from heterogeneous initial SoC values,
    all battery units converge to a common SoC trajectory over time. This demonstrates
    that the proposed distributed control law achieves SoC balancing among
    the battery units.

    \item Figure~\ref{fig:fig5} confirms that the total output power $p_{\Sigma}(t)$ closely
    tracks the desired total power $p^{*}(t)$ throughout the discharging
    process. This indicates that the global power demand is accurately satisfied while
    maintaining balanced energy utilization across the BESS.

    \item Figure~\ref{fig:fig6} illustrates the resulting power allocation $p_i(t)$ for each
    battery unit. The figure shows that the total power demand is distributed among the
    units in a coordinated manner based on their internal states, while no individual
    unit is required to reveal its private information to neighboring agents.

    \item Figures~\ref{fig:fig7} and~\ref{fig:fig8} show the evolution of the privacy-preserving
    distributed average unit state estimators $\hat{x}_{{\rm a},i}^{\alpha}(t)$ and
    $\hat{x}_{{\rm a},i}^{\beta}(t)$. The results demonstrate fast convergence of both
    sub-states to the scaled average unit state $\eta x_{\rm a}(t)$ with high accuracy,
    validating the effectiveness of the modified state-decomposition-based consensus
    mechanism.

    \item Figures~\ref{fig:fig9} demonstrates the behavior of the distributed estimator for
    the average desired power. The results show that each battery unit’s estimate
    $\hat{p}_{\mathrm{a},i}(t)$ closely follows the scaled signal $\sigma p_{\mathrm{a}}(t)$,
    confirming accurate estimation of the global power reference while preventing direct
    exposure of the true average desired power.
\end{itemize}

These results confirm that the proposed algorithm achieves the control objectives while
preserving privacy during discharging operation.

\subsection{Charging Mode Performance}

Figures~\ref{fig:fig10} to~\ref{fig:fig15} present the simulation results for the charging mode, where the desired charging power is defined as \( p^{*}(t)=(4200\sin(t) - 4200) \) W. The initial SoC values are set to \( (0.04, 0.11, 0.25, 0.2, 0.27, 0.12) \).

Similar to the discharging case, all battery units converge toward a common SoC
trajectory, and the total charging power accurately tracks the desired reference
$p^{*}(t)$. The individual charging power contributions are properly distributed among
the battery units, and the privacy-preserving estimators for both the average unit state
and the average desired power exhibit fast and accurate convergence. These results
demonstrate that the proposed algorithm maintains consistent performance in both
charging and discharging operating modes.
\begin{figure}[!t]
\centering
\includegraphics[width=4.8in, height=1.6in]{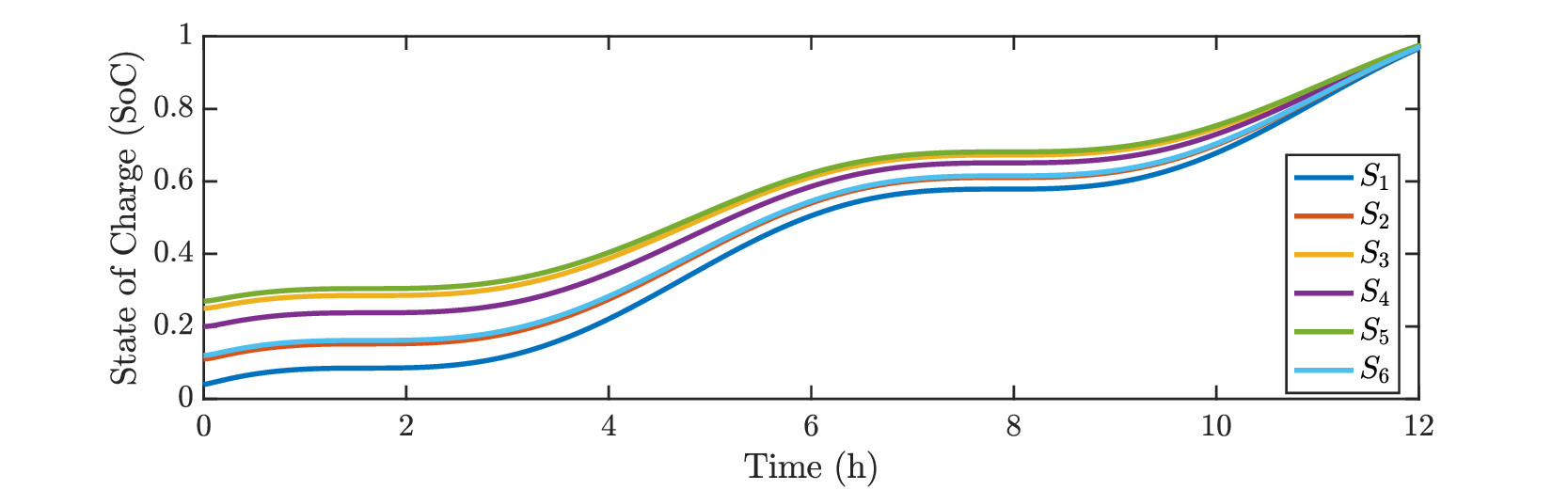}
\caption{The state of charge of each unit over time.}
\label{fig:fig10}
\end{figure}

\begin{figure}[!t]
\centering
\includegraphics[width=4.8in, height=1.6in]{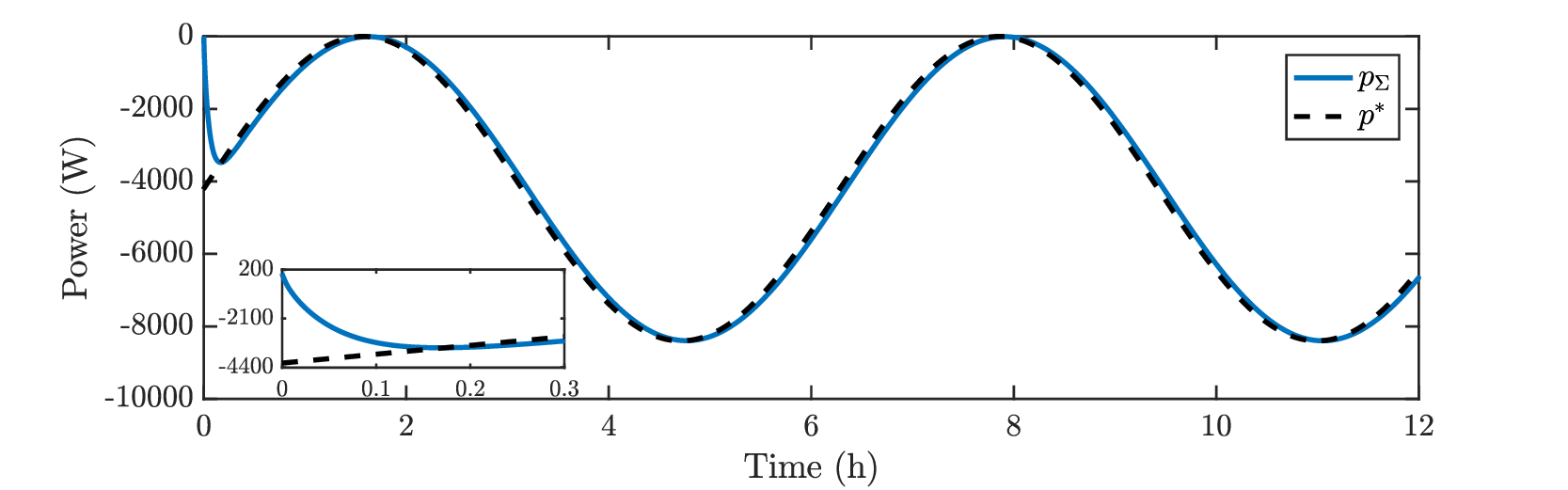}
\caption{The desired power and total charging power.}
\label{fig:fig11}
\end{figure}

\begin{figure}[!t]
\centering
\includegraphics[width=4.8in, height=1.6in]{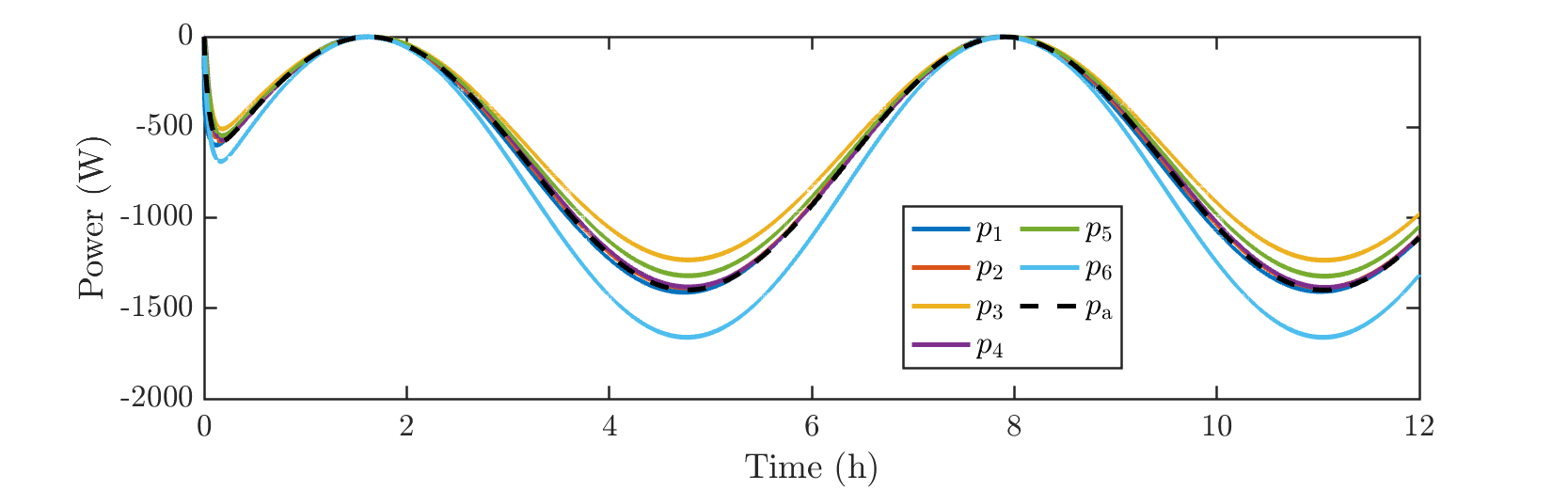}
\caption{The output power delivered by each battery.}
\label{fig:fig12}
\end{figure}

\begin{figure}[!t]
\centering
\includegraphics[width=4.8in, height=1.6in]{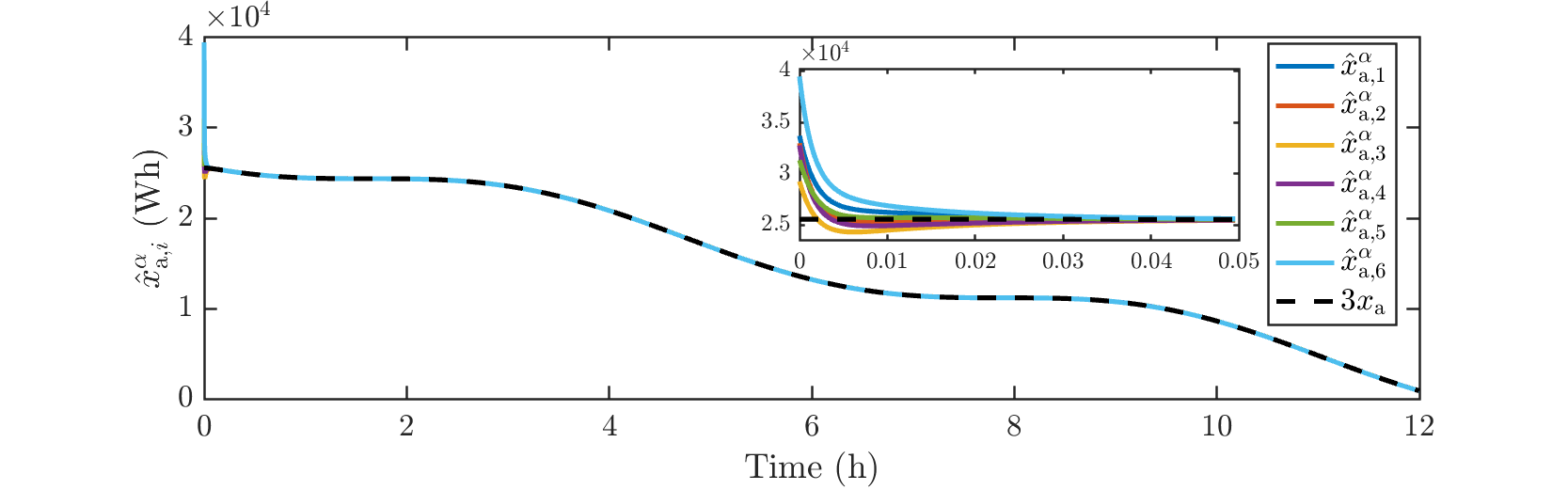}
\caption{The sub-state $\hat{x}_{{\rm a},i}^{\alpha}$ convergence to the scaled average unit state $\eta x_{\rm a}$.}
\label{fig:fig13}
\end{figure}

\begin{figure}[!t]
\centering
\includegraphics[width=4.8in, height=1.6in]{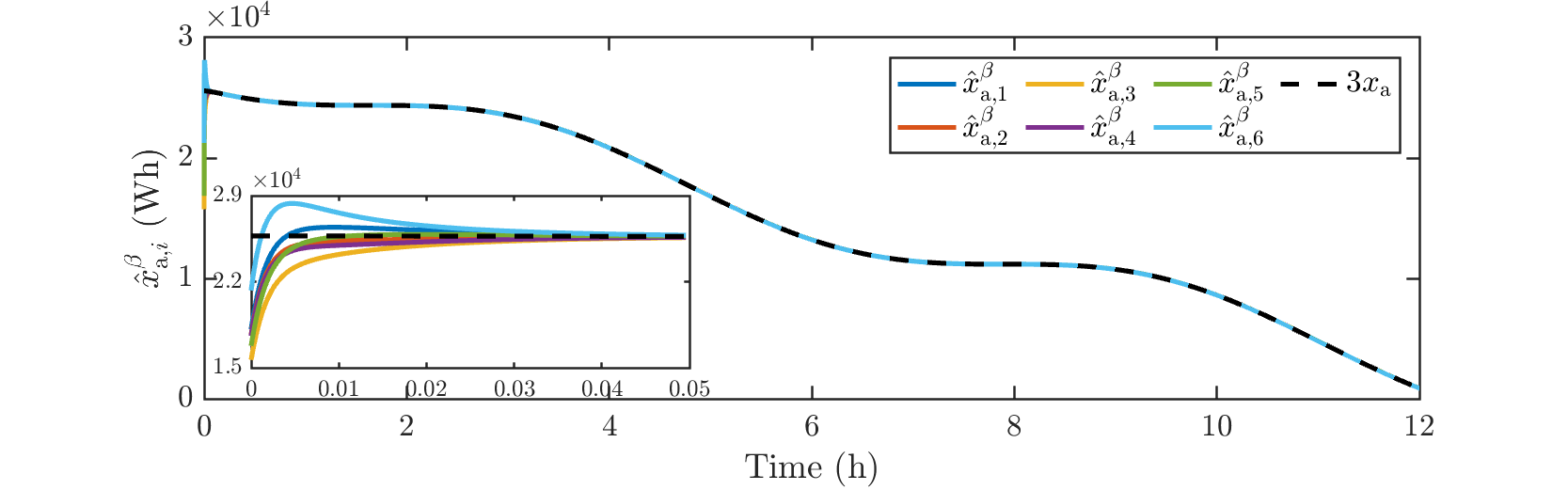}
\caption{The sub-state $\hat{x}_{{\rm a},i}^{\beta}$ convergence to the scaled average unit state $\eta x_{\rm a}$.}
\label{fig:fig14}
\end{figure}

\begin{figure}[!t]
\centering
\includegraphics[width=4.8in, height=1.6in]{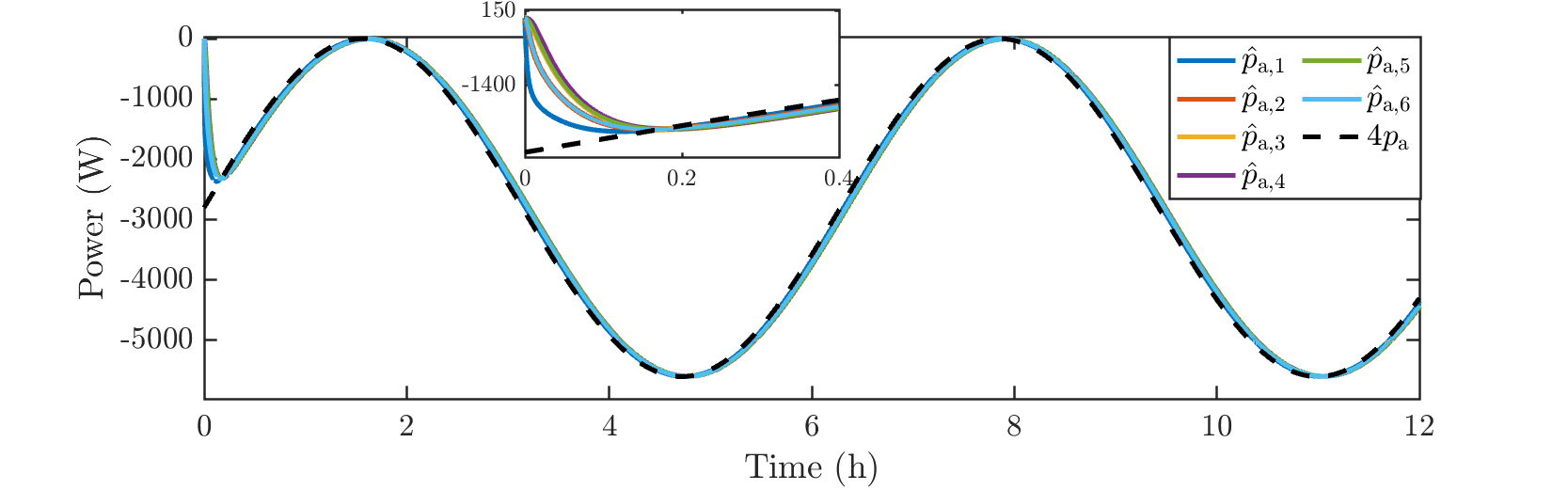}
\caption{The estimation of the scaled average desired power $\sigma p_{\rm a}$.}
\label{fig:fig15}
\end{figure}

\subsection{Privacy Verification Against External Eavesdropping}

Theorem~\ref{thm:4} establishes that the proposed algorithm prevents an external
eavesdropper from inferring the private states $x_i(t)$ and
$\dot{x}_i(t)$ for any admissible attack strategy. In this subsection, we will present simulation results illustratively validate this
theoretical guarantee under a representative observer-based eavesdropping
scenario.

Consider the discharging mode and suppose that an external eavesdropper attempts to infer the power delivery of each battery unit ($ p_i(t)=-\dot{x}_i(t)$) using the eavesdropping scheme described in Section~\ref{sec:attack}.

The simulation results under the external eavesdropping attack in the discharging mode
can be summarized as follows.
\begin{itemize}
    \item Figure~\ref{fig:fig16} presents the attacker’s reconstruction of the power
    delivered by each battery unit, $p_i(t) = -\dot{x}_i(t)$, when the non-privacy-preserving
    distributed control scheme \eqref{eq:10}, \eqref{eq:12}, and \eqref{eq:14} is employed.
    The reconstructed signals closely match the true power trajectories of all battery
    units, indicating that the attacker can accurately infer individual unit power
    profiles from the known and observed data.

    \item The results in Figure~\ref{fig:fig16} demonstrate that, in the absence of privacy
    protection, the internal information of each battery unit is fully
    exposed to an external eavesdropper, leading to a privacy breach at the
    individual unit level.

    \item Figure~\ref{fig:fig17} shows the attacker’s reconstruction results when the
    proposed privacy-preserving scheme described in \eqref{decomp},
    \eqref{eq:privacy desired power}, and \eqref{eq:18} is applied.
    In contrast to Fig.~\ref{fig:fig16}, the attacker’s estimates exhibit significant
    deviations from the true power trajectories $p_i(t)$ for all battery units.

    \item The results in Figure~\ref{fig:fig17} confirm that, although the proposed algorithm
    maintains accurate power allocation and SoC balancing (Figure~\ref{fig:fig5} and Figure~\ref{fig:fig4}), the individual battery unit
    power signals remain effectively protected from the external eavesdropper, thereby
    validating the privacy-preserving capability of the proposed framework.
\end{itemize}

\begin{figure}[!t]
\centering
\includegraphics[width=4.8in, height=1.6in]{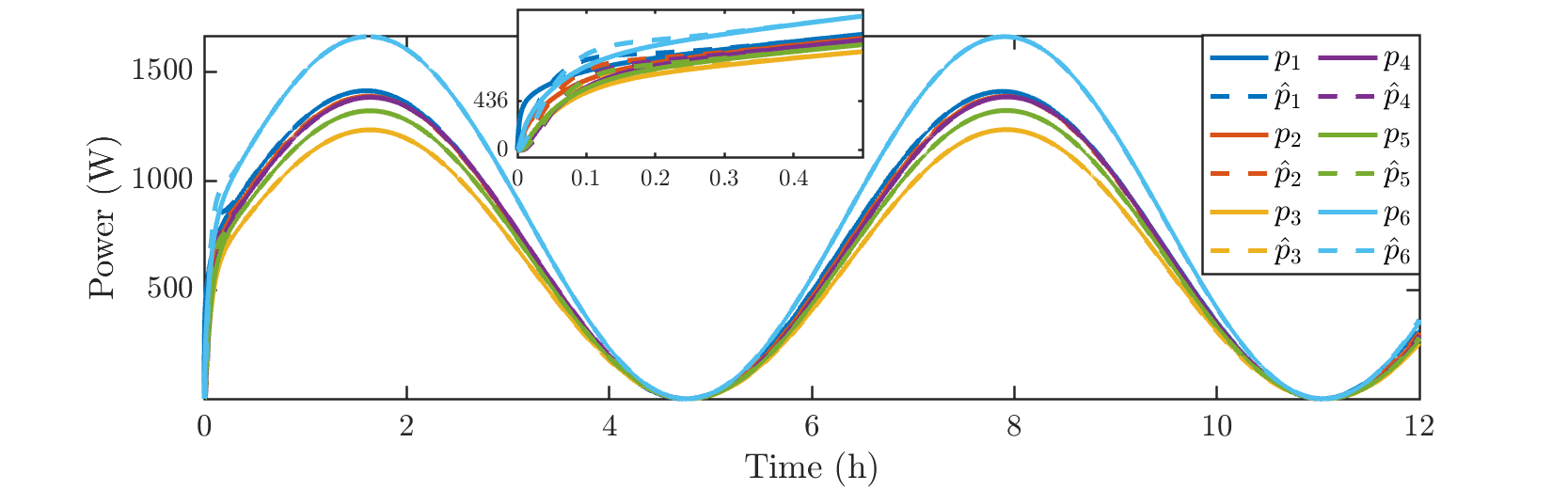}
\caption{Reconstruction of $p_i(t)$ based on the eavesdropping scheme described in Section~\ref{sec:attack} under the non-privacy-preserving scheme.}
\label{fig:fig16}
\end{figure}

\begin{figure}[!t]
\centering
\includegraphics[width=4.8in, height=1.6in]{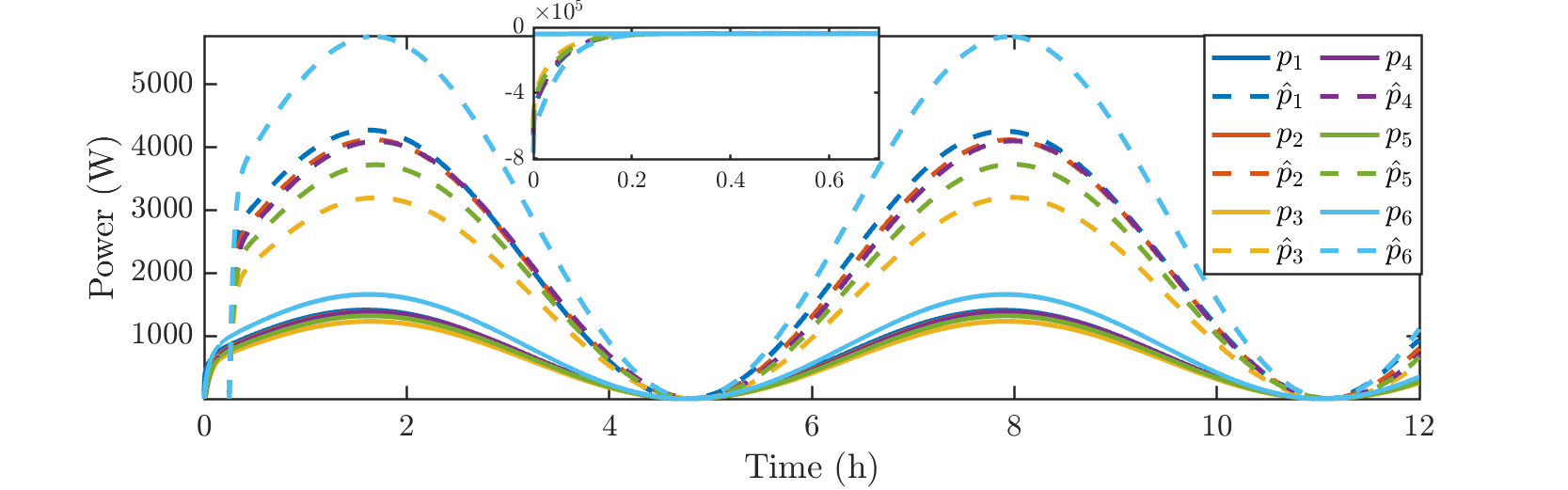}
\caption{Reconstruction of $p_i(t)$ based on the eavesdropping scheme described in Section~\ref{sec:attack} under the privacy-preserving scheme.}
\label{fig:fig17}
\end{figure}

\begin{figure}[!t]
\centering
\includegraphics[width=0.9\linewidth]{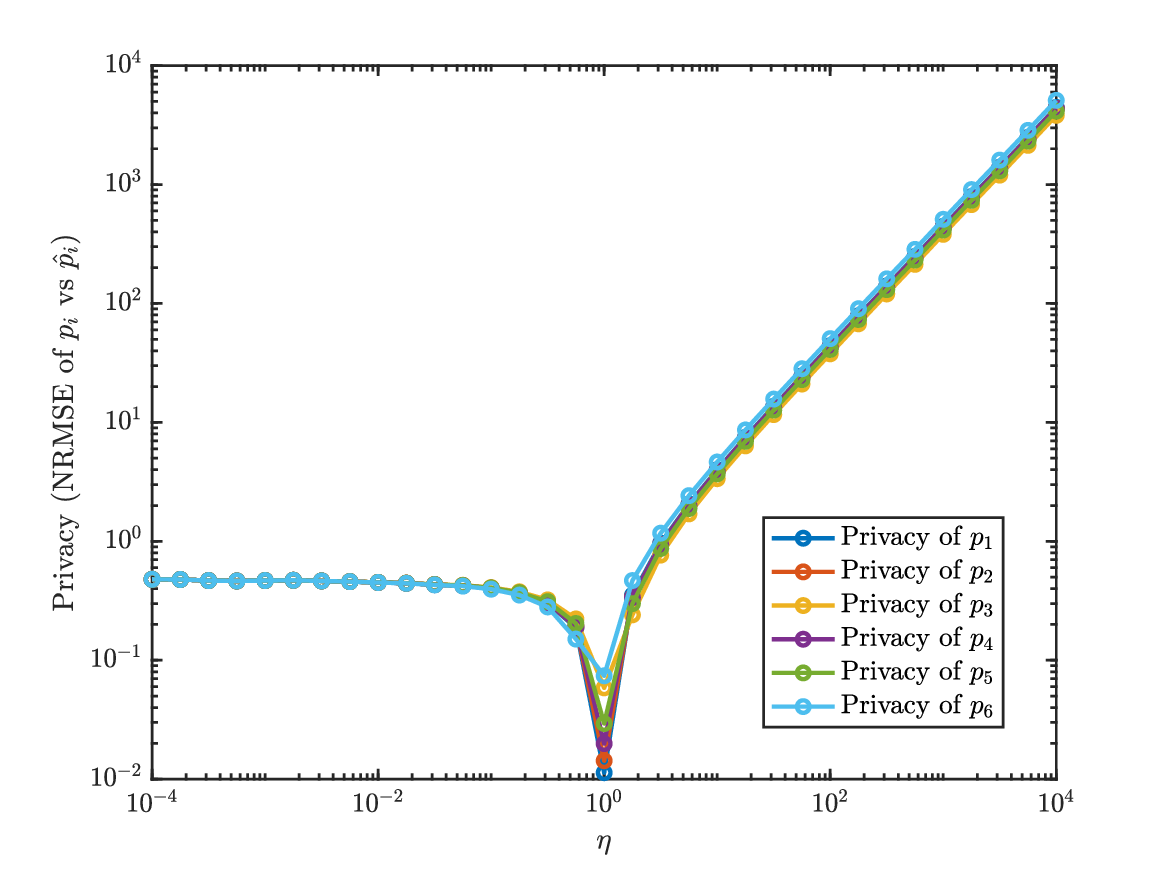}
\caption{Effect of $\eta$ on the privacy of $p_i$.}
\label{fig:fig18}
\end{figure}

\begin{figure}[!t]
\centering
\includegraphics[width=0.9\linewidth]{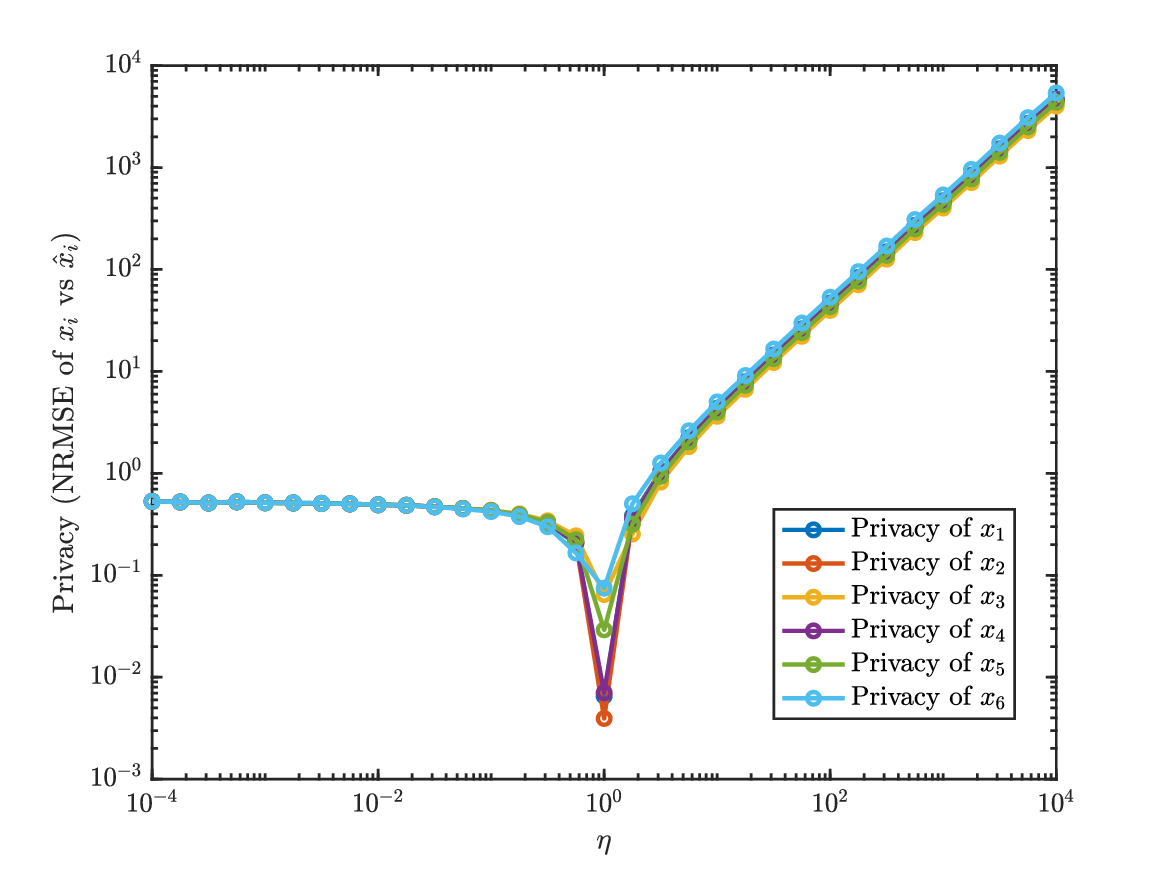}
\caption{Effect of $\eta$ on the privacy of $x_i$.}
\label{fig:fig21}
\end{figure}

\subsection{Impact of Parameter Selection on Privacy and System Performance}

Privacy of each battery unit's internal signals $x_i(t)$ and $p_i(t)=-\dot{x}_i(t)$ (for the discharging mode) is quantified with a normalized root mean square error (NRMSE). We compute the difference between the true signal and the attacker’s reconstructed signal, take the root mean square of this error over the post-transient interval, and then normalize it by the range of the true signal over the same interval. By design, a larger NRMSE indicates stronger privacy. Figure~\ref{fig:fig18} shows the attacker’s estimation error for the power delivered by each battery unit $p_i(t)$ over a range of $\eta$ values. The lowest privacy occurs at $\eta=1$, which corresponds to the conventional state-decomposition. For any $\eta>0$ with $\eta\neq 1$ (i.e., as $\eta$ moves away from $1$), the NRMSE increases, demonstrating that the proposed method provides stronger privacy than the conventional case.
The same qualitative dependence on $\eta$ is observed for the internal state $x_i(t)$, as shown in Figure~\ref{fig:fig21}, confirming consistent privacy behavior for both internal signals. Therefore, introducing $\eta$ preserves the privacy of not only the average reference input signal $x_{\mathrm{a}}(t)$ (as shown in Figure~\ref{fig:fig13}), but also the internal signals $p_i(t) = -\dot{x}_i(t)$ and $x_i(t)$.
 
As discussed in Theorem~\ref{thm:2} and Remark~\ref{remark1}, the scaling parameter $\eta$ does not affect the convergence speed or the steady-state convergence error of the proposed algorithm. Therefore, no trade-off exists between parameter selection, privacy protection, and system performance. As $\eta$ moves away from $\eta = 1$, privacy protection is enhanced without compromising convergence or control performance.

The parameter $\sigma$ protects the average desired power $p_{\rm{a}}(t)$ by scaling it before the leader transmits it to the other units. As shown in Figure~\ref{fig:fig9}, the external eavesdropper observes only a scaled version of the average desired power. According to Lemma~\ref{lem:2}, this scaling does not affect the convergence accuracy of the estimator. Therefore, choosing $\sigma$ away from $\sigma = 1$ provides privacy protection for $p_{\mathrm{a}}(t)$ without compromising system performance.

\subsection{Comparison Analysis}

In this subsection, we compare the proposed privacy-preserving distributed control algorithm with the case where the average unit state estimator is designed based on the state-decomposition-based method in \cite{zhang2022privacy}. We consider the discharging mode. For a fair comparison, identical parameter settings and network topology are used.

\begin{figure}[!t]
\centering
\includegraphics[width=4.8in, height=1.6in]{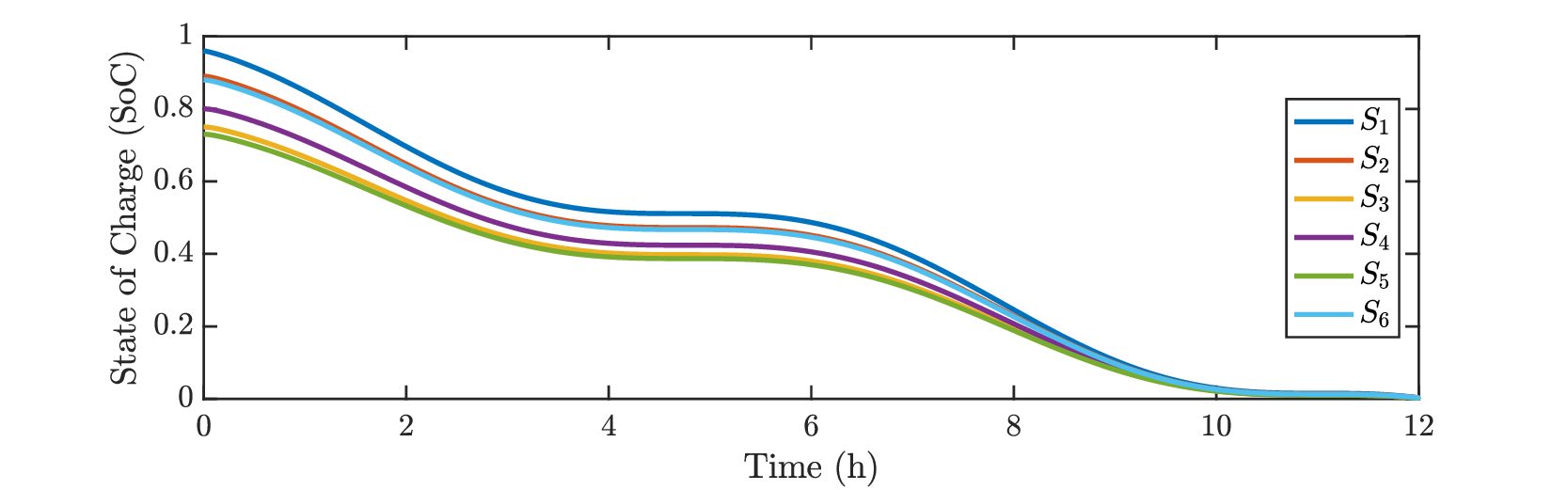}
\caption{The state of charge of each unit over time under the method in reference~\cite{zhang2022privacy}.}
\label{fig:fig19}
\end{figure}

\begin{figure}[!t]
\centering
\includegraphics[width=4.8in, height=1.6in]{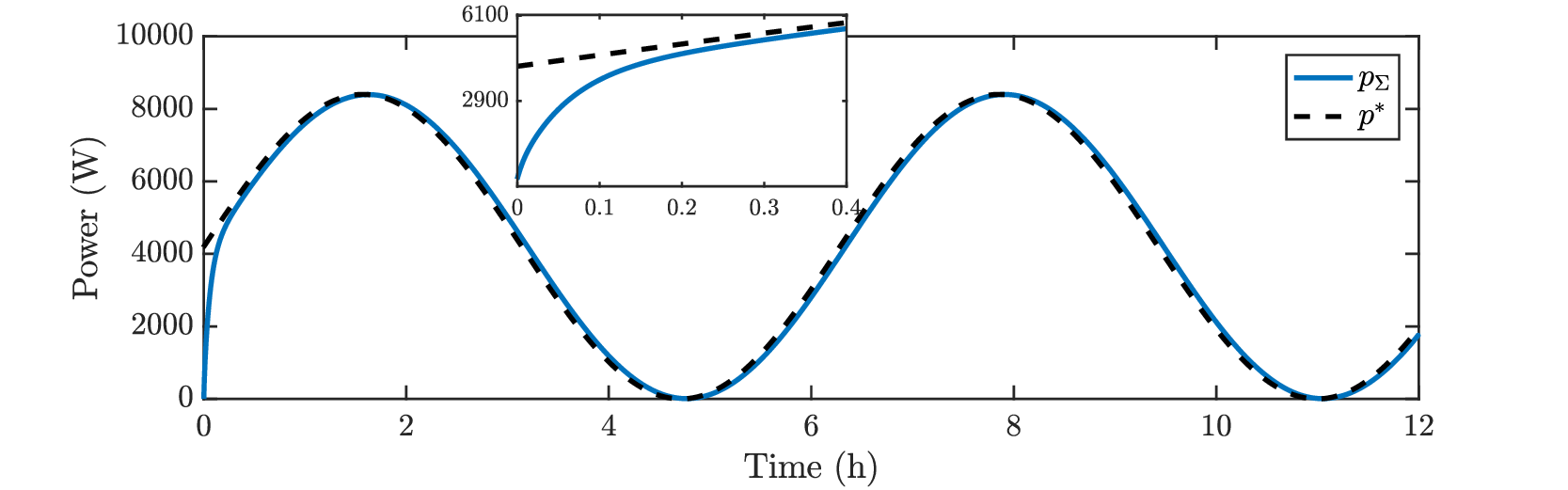}
\caption{The desired power and the total discharging power under the method in reference~\cite{zhang2022privacy}.}
\label{fig:fig20}
\end{figure}

Figures~\ref{fig:fig19} and~\ref{fig:fig20} illustrate the SoC balancing and power tracking performance of the distributed control algorithm when the average unit state estimator is designed based on the state-decomposition approach in~\cite{zhang2022privacy}. This demonstrates that the proposed privacy-preserving mechanism does not degrade the convergence accuracy or the control performance compared with the state-decomposition-based approach.

As discussed in Remark~\ref{remark2}, the method in \cite{zhang2022privacy} protects individual states but does not conceal the consensus value. In contrast, the proposed method intentionally scales the consensus trajectory itself, thereby preventing inference of both individual states and their average values.

The privacy effect is illustrated in Figure~\ref{fig:fig18} and Figure~\ref{fig:fig21}. When $\eta = 1$, the proposed method reduces to the state-decomposition-based approach in \cite{zhang2022privacy}, resulting in the lowest level of privacy.  Importantly, adjusting $\eta$ does not affect the convergence accuracy of the algorithm, which is consistent with the theoretical guarantees established in Theorem~\ref{thm:2} and Remark~\ref{remark1}.

Overall, the results demonstrate that the proposed method achieves comparable control performance while providing enhanced and tunable privacy protection compared with the state-decomposition-based approach.

\section{Conclusions}\label{sec:conclusion}

We proposed a novel privacy-preserving distributed control framework for SoC balancing and power tracking in a networked BESS. The framework integrates a power allocation law supported by two privacy-preserving distributed estimators, one for estimating the scaled average unit state and another for estimating the scaled average desired power. The average unit state estimator is constructed using a state-decomposition approach, which safeguards both individual battery states and their average. The average desired power estimator achieves privacy by coordinating over a scaled version of the true average desired power, thereby preventing external inference. Together, these estimators and the power allocation law enable accurate system coordination while preventing external eavesdroppers from inferring the internal states of individual batteries. Theoretical analysis and simulation results demonstrate the effectiveness, scalability, and ability to preserve privacy of the proposed approach, offering a promising solution for secure and efficient energy management in distributed power systems.

\section*{CRediT authorship contribution statement}

{\bf Mihitha Maithripala:} Writing – review \& editing, Writing – original
draft, Visualization, Methodology, Investigation, Formal analysis, Conceptualization. {\bf Zongli Lin:} Writing – review \& editing, Supervision, Project administration, Investigation, Funding acquisition, Formal
analysis, Conceptualization.

\section*{Declaration of competing interest}

The authors declare that they have no known competing financial interests or personal relationships that could have appeared to influence the work reported in this paper.

\section*{Data availability}

No data was used for the research described in the article.

\bibliographystyle{elsarticle-num}
\bibliography{reference-unmarked}

@article{yu2025optimized,
  title={Optimized distributed energy management for {BESS} incorporating time-varying delays with an improved bipartite grouping model simultaneously balancing {SoH} and {SoC}},
  author={Yu, Yang and Wang, Boxiao and Li, Menglu and Lv, Tingyan},
  journal={Journal of Energy Storage},
  volume={126},
  pages={117006},
  year={2025},
  publisher={Elsevier}
}

@article{hill2012battery,
  title={Battery energy storage for enabling integration of distributed solar power generation},
  author={Hill, C A and Such, M C and Chen, D and Gonzalez, J and Grady, W M},
  journal={IEEE Transactions on Smart Grid},
  volume={3},
  number={2},
  pages={850--857},
  year={2012}
}

@article{lawder2014battery,
  title={Battery energy storage system ({BESS}) and battery management system ({BMS}) for grid-scale applications},
  author={Lawder, Matthew T and Northrop, Paul WC and Subramanian, Venkat R and Braatz, Richard D and Klein, Richard and Howey, David A},
  journal={Proceedings of the IEEE},
  volume={102},
  number={6},
  pages={1014--1030},
  year={2014}
}

@article{xu2018distributed,
  title={Distributed robust control strategy of grid-connected inverters for energy storage systems’ state-of-charge balancing},
  author={Xu, Y and Liu, W and Liu, Y},
  journal={IEEE Transactions on Smart Grid},
  volume={9},
  number={6},
  pages={5907--5917},
  year={2018}
}

@article{qiu2025distributed,
  title={Distributed optimization for the trade-off between state-of-charge balancing and strategic battery usage in a networked {BESS}},
  author={Qiu, Chenyang and Qian, Yangyang and Lin, Zongli and Shamash, Yacov A},
  journal={Journal of Energy Storage},
  volume={105},
  pages={114550},
  year={2025},
  publisher={Elsevier}
}

@article{xing2019distributed,
  title={Distributed state-of-charge balance control with event-triggered signal transmissions for multiple energy storage systems in smart grid},
  author={Xing, Ling and Mishra, Y and Tian, Y-C and Ledwich, G and Zhou, C and Du, W and Qian, F},
  journal={IEEE Transactions on Systems, Man, and Cybernetics: Systems},
  volume={49},
  number={8},
  pages={1601--1612},
  year={2019}
}

@article{liu2025fixed,
  title={Fixed-time quasi-consensus energy management method for battery energy storage systems in {DC} microgrids under two types of {DoS} attacks},
  author={Liu, Dan and Jiang, Kezheng and Xiong, Ping and Tian, Xu and Wang, Rui and Sun, Qiuye},
  journal={Journal of Energy Storage},
  volume={113},
  pages={115574},
  year={2025},
  publisher={Elsevier}
}

@article{lu2014state,
  title={State-of-charge balance using adaptive droop control for distributed energy storage systems in {DC} microgrid applications},
  author={Lu, Xiaonan and Sun, Kai and Guerrero, Josep M and Vasquez, Juan C and Huang, Ling},
  journal={IEEE Transactions on Industrial Electronics},
  volume={61},
  number={6},
  pages={2804--2815},
  year={2014}
}

@article{qu2018cooperative,
  title={Cooperative control of heterogeneous connected energy storage systems for dynamic load sharing in microgrids},
  author={Qu, Zhihua and Wang, Jian},
  journal={IEEE Transactions on Industrial Electronics},
  volume={65},
  number={8},
  pages={6541--6551},
  year={2018}
}

@article{bidram2012secondary,
  title={Secondary control of microgrids based on distributed cooperative control of multi-agent systems},
  author={Bidram, Ali and Davoudi, Ali and Lewis, Frank L and Qu, Zhihua},
  journal={IET Generation, Transmission and Distribution},
  volume={7},
  number={8},
  pages={822--831},
  year={2012}
}

@article{huang2023voltage,
  title={A voltage-shifting-based state-of-charge balancing control for distributed energy storage systems in islanded {DC} microgrids},
  author={Huang, Zuliang and Li, Yan and Cheng, Xin and Ke, Mingjun},
  journal={Journal of Energy Storage},
  volume={69},
  pages={107861},
  year={2023},
  publisher={Elsevier}
}

@article{yan2019event,
  title={Event-triggered distributed control for {SoC} balancing of distributed energy storage systems in {DC} microgrids},
  author={Yan, Yong and Yue, Dong and Pang, Chao and Yang, Jun},
  journal={IEEE Transactions on Smart Grid},
  volume={10},
  number={6},
  pages={6613--6623},
  year={2019}
}

@article{kossek2024privacy,
  title={Privacy in cooperative control of cyber-physical systems: A survey of techniques and challenges},
  author={Kossek, Matej and Stefanovic, Milutin},
  journal={Journal of Intelligent and Robotic Systems},
  volume={110},
  number={129},
  pages={1--27},
  year={2024}
}

@article{kia2015dynamic,
  title={Dynamic average consensus under limited control authority and privacy requirements},
  author={Kia, S S and Cort{\'e}s, J and Mart{\'\i}nez, S},
  journal={International Journal of Robust and Nonlinear Control},
  volume={25},
  number={13},
  pages={1941--1966},
  year={2015}
}

@inproceedings{7852360,
  title={Distributed computing over encrypted data},
  author={Freris, Nikolaos M and Patrinos, Panagiotis},
  booktitle={2016 54th Annual Allerton Conference on Communication, Control, and Computing (Allerton)},
  pages={1116--1122},
  year={2016},
  organization={IEEE}
}

@article{lu2018privacy,
  title={Privacy preserving distributed optimization using homomorphic encryption},
  author={Lu, Yang and Zhu, Minghui},
  journal={Automatica},
  volume={96},
  pages={314--325},
  year={2018},
  publisher={Elsevier}
}

@article{Nozari2015DifferentiallyPA,
  title={Differentially private average consensus: Obstructions, trade-offs, and optimal algorithm design},
  author={Nozari, Erfan and Tallapragada, Pavankumar and Cort{\'e}s, Jorge},
  journal={Automatica},
  volume={81},
  pages={221--231},
  year={2017},
  publisher={Elsevier}
}

@article{ding2021differentially,
  title={Differentially private distributed optimization via state and direction perturbation in multiagent systems},
  author={Ding, Tie and Zhu, Shanying and He, Jianping and Chen, Cailian and Guan, Xinping},
  journal={IEEE Transactions on Automatic Control},
  volume={67},
  number={2},
  pages={722--737},
  year={2021},
  publisher={IEEE}
}

@ARTICLE{9910413,
  author={Zhao, Daduan and Zhang, Chenghui and Cao, Xiangyang and Peng, Chao and Sun, Bo and Li, Ke and Li, Yan},
  journal={IEEE Transactions on Control Systems Technology}, 
  title={Differential Privacy Energy Management for Islanded Microgrids With Distributed Consensus-Based {ADMM} Algorithm}, 
  year={2023},
  volume={31},
  number={3},
  pages={1018-1031},
  }

@article{inoue2007bidirectional,
  title={A Bidirectional {DC--DC} Converter for an Energy Storage System With Galvanic Isolation},
  author={Inoue, Shigenori and Akagi, Hirofumi},
  journal={IEEE Transactions on Power Electronics},
  year={2007},
  volume={22},
  pages={2299--2306},
  
}

@article{meng2021distributed,
  author={Meng, Tingyang and Lin, Zongli and Shamash, Yacov A.},
  journal={IEEE/CAA Journal of Automatica Sinica}, 
  title={Distributed Cooperative Control of Battery Energy Storage Systems in {DC} Microgrids}, 
  year={2021},
  volume={8},
  number={3},
  pages={606--616},
  
}

@article{qian2023distributed,
  title={Distributed event-triggered algorithms for the management of networked battery systems},
  author={Qian, Yangyang and Meng, Tingyang and Lin, Zongli and Wan, Yan and Shamash, Yacov A},
  journal={International Journal of Robust and Nonlinear Control},
  volume={35},
  number={17},
  pages={7041--7066},
  year={2025},
  publisher={Wiley Online Library}
}

@inproceedings{spanos2005dynamic,
  title={Dynamic consensus on mobile networks},
  author={Spanos, Demetri P and Olfati-Saber, Reza and Murray, Richard M},
  booktitle={IFAC world congress},
  pages={1--6},
  year={2005}
}

@article{kia2019tutorial,
  author={Kia, Solmaz S. and Van Scoy, Bryan and Cortes, Jorge and Freeman, Randy A. and Lynch, Kevin M. and Martinez, Sonia},
  journal={IEEE Control Systems Magazine}, 
  title={Tutorial on Dynamic Average Consensus: The Problem, Its Applications, and the Algorithms}, 
  year={2019},
  volume={39},
  number={3},
  pages={40--72},
 
}

@ARTICLE{7548310,
  author={Cai, He and Hu, Guoqiang},
  journal={IEEE Transactions on Industrial Informatics}, 
  title={Distributed Control Scheme for Package-Level State-of-Charge Balancing of Grid-Connected Battery Energy Storage System}, 
  year={2016},
  volume={12},
  number={5},
  pages={1919-1929},
 }

@article{zhang2022privacy,
  title={Privacy-preserving dynamic average consensus via state decomposition: Case study on multi-robot formation control},
  author={Zhang, Kaixiang and Li, Zhaojian and Wang, Yongqiang and Louati, Ali and Chen, Jian},
  journal={Automatica},
  volume={139},
  pages={110182},
  year={2022},
  publisher={Elsevier}
}

@ARTICLE{8657789,
  author={Wang, Yongqiang},
  journal={IEEE Transactions on Automatic Control}, 
  title={Privacy-Preserving Average Consensus via State Decomposition}, 
  year={2019},
  volume={64},
  number={11},
  pages={4711-4716},
 }

@article{luan2024privacy,
  title={Privacy-Preserving Optimization Algorithm for Distributed Energy Management Over Time-Varying Graphs: A State Decomposition Method},
  author={Luan, Meng and Wen, Guanghui and Yang, Tao},
  journal={International Journal of Robust and Nonlinear Control},
  volume={35},
  number={17},
  pages={7190--7201},
  year={2025},
  publisher={Wiley Online Library}
}

@ARTICLE{4118472,
  author={Olfati-Saber, Reza and Fax, J. Alex and Murray, Richard M.},
  journal={Proceedings of the IEEE}, 
  title={Consensus and Cooperation in Networked Multi-Agent Systems}, 
  year={2007},
  volume={95},
  number={1},
  pages={215-233},
  }

@article{ruan2019secure,
  title={Secure and privacy-preserving consensus},
  author={Ruan, Minghao and Gao, Huan and Wang, Yongqiang},
  journal={IEEE Transactions on Automatic Control},
  volume={64},
  number={10},
  pages={4035--4049},
  year={2019},
  publisher={IEEE}
}

@article{chen2022privacy,
  title={Privacy-preserving distributed economic dispatch of microgrids: A dynamic quantization-based consensus scheme with homomorphic encryption},
  author={Chen, Wei and Liu, Lu and Liu, Guo-Ping},
  journal={IEEE Transactions on Smart Grid},
  volume={14},
  number={1},
  pages={701--713},
  year={2022},
  publisher={IEEE}
}

@article{chen2024new,
  title={A new privacy-preserving average consensus algorithm with two-phase structure: Applications to load sharing of microgrids},
  author={Chen, Wei and Wang, Zidong and Liu, Qinyuan and Yue, Dong and Liu, Guo-Ping},
  journal={Automatica},
  volume={167},
  pages={111715},
  year={2024},
  publisher={Elsevier}
}

\end{document}